\documentclass[11pt,onecolumn]{article}
\usepackage[margin=1in]{geometry}
\usepackage{amsmath,amssymb,amsthm}
\usepackage{graphicx}
\usepackage{bm}
\usepackage{mathtools}
\usepackage{float}
\usepackage{placeins}
\usepackage{booktabs}
\usepackage{array}
\usepackage{microtype}
\usepackage[colorlinks=true,linkcolor=blue,citecolor=blue,urlcolor=blue]{hyperref}
\hypersetup{
  pdftitle={What solvable models reveal about Born--Oppenheimer, Born--Huang, and exact factorization},
  pdfauthor={Stephen Wiggins},
  pdfsubject={Born--Oppenheimer, Born--Huang, exact factorization, and solvable molecular benchmarks},
  pdfkeywords={Born--Oppenheimer approximation, Born--Huang expansion, exact factorization, quantum metric, derivative coupling}
}

\newcommand{\R}{\mathbf{R}}
\newcommand{\rr}{\mathbf{r}}
\newcommand{\bra}[1]{\langle #1 |}
\newcommand{\ket}[1]{| #1 \rangle}
\newcommand{\bk}[2]{\langle #1 | #2 \rangle}
\newcommand{\Tn}{\hat{T}_{\mathrm{nuc}}}
\newcommand{\Hel}{\hat{H}_{\mathrm{el}}}

\newcommand{\Ecal}{\mathcal{E}}

\newtheorem{proposition}{Proposition}

\title{\bf What solvable models reveal about\\
Born--Oppenheimer, Born--Huang, and exact factorization}

\author{Stephen Wiggins\\[2pt]
\small Hetao Institute of Mathematics and Interdisciplinary Science, Shenzhen, China\\
\small University of Bristol, Bristol, United Kingdom}

\date{}

\begin{document}
\maketitle

\begin{abstract}
\noindent
The Born--Oppenheimer (BO) and single-surface Born--Huang (BH) approximations replace the full molecular problem by nuclear motion associated with one clamped electronic state, whereas the full BH expansion and exact factorization (EF) are exact representations. We use two analytically transparent models to identify what controls one-surface error. For Fern\'andez's bilinearly coupled oscillators, the Brattsev ground-state bracket was already known; we derive closed-form differences proving strictness for every positive nuclear mass and nonzero stable coupling. Using Hunter's conditional-amplitude construction, we place the closed-form ground-state EF potential alongside BO and single-surface BH. Exact excited-state calculations show that the bracket does not extend through the spectrum. The diagonal BH correction is the quantum-metric coefficient of the retained electronic state divided by twice the nuclear mass. Fern\'andez's sixth-order expansion shows that the next omitted correction contains the same derivative couplings, an additional inverse electronic-energy separation, and a factor depending on the nuclear state. We therefore add a two-state model with a position-dependent gap. Decreasing its minimum separation makes the metric higher and narrower while leaving the total Fubini--Study length equal to $\pi/2$. The exact EF, BO, and BH potentials can differ visibly even when their nuclear ground-state amplitudes are nearly identical. The two models show that one-surface accuracy depends on nuclear mass, electronic-state variation, energy separation from omitted states, and the nuclear wavefunction.
\end{abstract}


\section{Introduction}
\label{sec:intro}

What, precisely, is the potential energy surface on which nuclei move?  The exact
nonrelativistic molecular Schr\"odinger equation describes electrons and nuclei together.  In
principle, an exact solution of that equation would already give the molecular energy and
wavefunction, and no separate potential for the nuclei would be needed.  For realistic
molecules, however, solving the full electron--nuclear problem directly is not computationally
tractable on present classical computers.  The Born--Oppenheimer construction begins by
holding the nuclear positions fixed, solving an electronic Schr\"odinger equation, and then
using the resulting electronic energies in a calculation of nuclear motion~\cite{BO1927}.
The potential-energy surface is therefore not an additional ingredient put into the exact
molecular problem at the start; it appears when one asks for a nuclear description.  Recent
pedagogical discussions emphasize the same logical order and, in particular, that fixing the
nuclei while solving the electronic problem is a step in the construction rather than the
Born--Oppenheimer approximation itself~\cite{AgostiniCurchod2026,PrljEtAl2026}.

There is more than one way to obtain such a nuclear description.  In the Born--Huang
construction, the electronic states calculated at fixed nuclear geometry are used as a basis
for the molecular wavefunction~\cite{BornHuang1954}.  If all of those electronic states are
retained, this is an exact representation of the molecular problem.  Restricting the
representation to one electronic state produces a one-surface approximation.  The strict
Born--Oppenheimer approximation used in this paper omits the terms that arise when the
nuclear kinetic-energy operator differentiates that electronic state.  The corresponding
single-surface Born--Huang approximation keeps the derivative contribution that acts within
the retained electronic state while still omitting terms that couple it to the other
electronic states.  We will derive these statements rather than assume them.

Exact factorization provides a different exact description of the same molecular state.  It
writes the molecular wavefunction as one nuclear amplitude multiplied by one normalized
electronic factor~\cite{Hunter1974,Hunter1975,Cederbaum2013,AbediMaitraGross2010,AbediMaitraGross2012}.
The relation to the full Born--Huang representation is especially simple.  Once the
Born--Huang coefficients have been introduced, the same elementary linear-algebra operation
can describe them either as components in a basis or as a length and a normalized direction.
Sch\"urger, Lassmann, Agostini, and Curchod recently made this connection explicit using the
Born--Huang coefficient vector~\cite{SchurgerLassmannAgostiniCurchod2025}.
Section~\ref{sec:hierarchy} will develop this
relation only after all of the objects entering it have been defined.

The distinctions are easy to state but difficult to isolate in realistic molecular
calculations.  The electronic problem must itself be approximated, only a finite set of
electronic states can usually be retained, the nuclear equations must be discretized and
converged, and a complete electronic description may require continuum states.  An observed
error can therefore mix several effects.  The purpose of the present paper is
narrower.  We use models simple enough that the full molecular problem and the relevant
one-surface approximations can be placed side by side, in several cases with exact
calculations.  Removing the other complications lets us ask what controls the error caused
specifically by restricting the electron--nuclear problem to one electronic state.

The first benchmark is the bilinearly coupled-oscillator model used by Fern\'andez in his
sixth-order study of corrections to the Born--Oppenheimer approximation
\cite{Fernandez1994}.  His expansion extends the earlier isotope- and harmonic-oscillator
correction analyses of Van Vleck and Patterson~\cite{VanVleck1936,Patterson1993}.  Exact
solution of coupled harmonic oscillators by normal coordinates is standard
~\cite{ZunigaBastidaRequena2017}, and McCoy later used the same bilinear oscillator family
to examine the corresponding one-surface approximation and its ground-state error
~\cite{McCoy2001}.
The value of the model here is not its novelty but its transparency.  From one Hamiltonian we
can calculate the exact molecular energies, the Born--Oppenheimer and single-surface
Born--Huang energies, the change of the electronic states with nuclear position, and the
ground-state exact-factorization potential.  Fern\'andez already stated and verified the
Brattsev ground-state bracket for this oscillator~\cite{Fernandez1994,Brattsev1965}.  The closed-form differences derived below
show directly that the bracket is strict for every positive nuclear mass and every nonzero
stable coupling, without restricting the statement to the range of his mass-ratio series; exact
excited states show that the same ordering is not an excited-state rule.

The derivative terms in the Born--Huang calculation also contain geometric information, but
we postpone the geometric language until the derivative itself has been derived.  The basic
question is elementary: when a normalized electronic state changes as the nuclei move, how
much of that change is a physically different electronic state and how much is only a change
of phase?  The squared size of the physically different part is the quantum, or
Fubini--Study, metric.  Once that quantity has been constructed, the diagonal Born--Huang
correction turns out to be the metric divided by twice the nuclear mass in one dimension.
This relation is standard and has been emphasized in recent molecular quantum-geometry work
\cite{ProvostVallee1980,KronenbergerSteinmetzAppenzellerPausch2025,XieLiuGu2025}.  The
geometric terminology is introduced in Section~\ref{sec:metric} only after the calculation
that gives it meaning.

The Fern\'andez model has one limitation: its electronic energy gaps are constant.
We therefore add a two-state model in which the electronic energy separation changes with
nuclear position.  The model is intentionally minimal: its $2\times2$ electronic matrix can
be diagonalized explicitly, so we can see exactly how changing the minimum separation
changes the electronic eigenvectors and the corresponding derivative terms.  The same model
then lets us compare the Born--Oppenheimer, single-surface Born--Huang, and
exact-factorization nuclear descriptions numerically.

The paper follows the order in which the ideas are needed.  Section~\ref{sec:hierarchy}
starts from the exact molecular equation, constructs the fixed-nuclei electronic problem,
and then derives the full Born--Huang representation, the two one-surface approximations, and
the elementary components-versus-length-and-direction relation to exact factorization.
Section~\ref{sec:fernandez} develops the Fern\'andez benchmark.  Section~\ref{sec:metric}
returns to the electronic derivatives already encountered and separates the part changed by a
phase choice from the perpendicular change of the electronic state.  Section~\ref{sec:needEF} then applies the same construction to the
normalized electronic factor in exact factorization.  Sections~\ref{sec:vibronic} and
\ref{sec:EFvib} study the variable-gap model.  Section~\ref{sec:UQ} finally uses the two
benchmarks to identify the quantities that control the error of a one-surface description.

\section{From the molecular Schr\"odinger equation to nuclear descriptions}
\label{sec:hierarchy}

The question in this section is: starting with the exact molecular Schr\"odinger equation,
what calculation produces each of the nuclear descriptions used later?

\subsection{The exact electron--nuclear problem}

We begin with the standard nonrelativistic molecular Hamiltonian.  A general and accessible
account of the molecular quantum mechanics used here is given by Atkins and Friedman
\cite{AtkinsFriedman2011}.

Let $\rr=(\mathbf r_1,\ldots,\mathbf r_{N_{\mathrm e}})$ denote the electronic coordinates
and $\R=(\mathbf R_1,\ldots,\mathbf R_{N_{\mathrm n}})$ the nuclear coordinates.  Let
$\hat H$ denote the molecular Hamiltonian and $E$ a stationary molecular energy.  A
stationary molecular state is a wavefunction $\Psi(\rr,\R)$ satisfying
\begin{equation}
\hat H\Psi(\rr,\R)=E\Psi(\rr,\R).
\label{eq:fullSE}
\end{equation}
For each nucleus $\alpha$, let $M_\alpha$ be its mass and let $\nabla_\alpha$ denote
differentiation with respect to $\mathbf R_\alpha$.  In atomic units, let $\Tn$ denote
the nuclear kinetic-energy operator and $\Hel$ the remaining electronic Hamiltonian.  Then
\begin{equation}
\hat H=\Tn+\Hel,
\qquad
\Tn=-\sum_{\alpha=1}^{N_{\mathrm n}}
\frac{1}{2M_\alpha}\nabla_\alpha^2.
\label{eq:Hsplit}
\end{equation}
Let $V(\rr,\R)$ denote the potential energy of all Coulomb interactions.  The electronic
Hamiltonian is
\begin{equation}
\Hel(\rr;\R)=
-\sum_{i=1}^{N_{\mathrm e}}\frac12\nabla_{\mathbf r_i}^2+V(\rr,\R).
\label{eq:Hel}
\end{equation}
Let $Z_\alpha$ denote the charge of nucleus $\alpha$.  For Coulomb particles the potential
in Eq.~\eqref{eq:Hel} is
\begin{equation}
V(\rr,\R)=
\sum_{i<j}\frac{1}{|\mathbf r_i-\mathbf r_j|}
-\sum_{i,\alpha}\frac{Z_\alpha}{|\mathbf r_i-\mathbf R_\alpha|}
+\sum_{\alpha<\beta}\frac{Z_\alpha Z_\beta}{|\mathbf R_\alpha-\mathbf R_\beta|}.
\label{eq:coulombV}
\end{equation}
Nothing has yet been separated: the unknown in Eq.~\eqref{eq:fullSE} is one function of all
electronic and nuclear coordinates.  The spectral questions that arise for the Coulomb
Hamiltonian when one later introduces fixed-nuclei electronic states have been analyzed
carefully by Sutcliffe and Woolley \cite{SutcliffeWoolley2012,SutcliffeWoolley2014}.  Our
benchmark models avoid those complications.

\subsection{First construct the fixed-nuclei electronic states}

Choose a nuclear geometry $\R$ and hold it fixed.  The operator $\Hel(\rr;\R)$ then acts only
on the electronic variables.  Let $\phi_k(\rr;\R)$ denote its $k$th normalized electronic
eigenfunction and $E_k(\R)$ the corresponding eigenvalue.  They are defined by
\begin{equation}
\Hel(\rr;\R)\phi_k(\rr;\R)=E_k(\R)\phi_k(\rr;\R).
\label{eq:clamped}
\end{equation}
The label $k$ distinguishes electronic eigenstates.  Repeating this calculation as $\R$
changes gives an electronic eigenvalue $E_k(\R)$ and an electronic eigenfunction
$\phi_k(\rr;\R)$ at each nuclear geometry.  When a nuclear equation is later restricted to
this electronic state, the function $E_k(\R)$ supplies the potential-energy term in that
nuclear equation.  Constructing the fixed-nuclei electronic states is not yet the
Born--Oppenheimer approximation; it provides the electronic states from which that
approximation, and the Born--Huang expansion, are built \cite{AgostiniCurchod2026}.

At each fixed $\R$, the electronic eigenstates form a basis, or more generally a spectral
resolution, for functions of the electronic coordinates.  Let $\chi_k(\R)$ denote the
coefficient of $\phi_k$ when the molecular wavefunction is expanded in that basis.  Then
\begin{equation}
\Psi(\rr,\R)=\sum_k\chi_k(\R)\phi_k(\rr;\R).
\label{eq:BHsum}
\end{equation}
The functions $\chi_k(\R)$ are the expansion coefficients; they depend only on the nuclear
coordinates.  Equation~\eqref{eq:BHsum} is the Born--Huang expansion
\cite{BornHuang1954}.  If the electronic spectral resolution is complete, it is an exact
rewriting of the molecular state.  For a Coulomb electronic Hamiltonian, completeness may
require continuum contributions as well as discrete states
\cite{SutcliffeWoolley2012,SutcliffeWoolley2014}.  No continuum occurs in the finite
benchmark models below.

\subsection{Two exact descriptions of the same molecular state}
\label{sec:EFconstruction}

We now have enough notation to make the relation between the full Born--Huang expansion and
exact factorization explicit.  Sch\"urger, Lassmann, Agostini, and Curchod recently obtained
the exact-factorization electronic state directly from the Born--Huang coefficient vector:
the rank-one density matrix formed from those coefficients has one nonzero eigenvalue, and
its normalized eigenvector is proportional to the coefficient vector itself
\cite{SchurgerLassmannAgostiniCurchod2025}.  We use the equivalent elementary vector
language because it makes the same relation visible without first introducing a density
matrix.

At fixed $\R$, the molecular wavefunction is a function only of the electronic coordinates.
We use ket notation for this electronic vector and denote it by $\ket{\Psi_\R}$:
\begin{equation}
\ket{\Psi_\R}\equiv \Psi(\,\cdot\,,\R).
\label{eq:PsiRvector}
\end{equation}
Equation~\eqref{eq:BHsum} describes this vector by its components in the fixed-nuclei
electronic basis,
\begin{equation}
\ket{\Psi_\R}=\sum_k\chi_k(\R)\ket{\phi_k(\R)}.
\label{eq:introBH}
\end{equation}
Nothing has been approximated.  This is the familiar description of a vector by its
components in a chosen basis.

The same nonzero vector can instead be described by its length and a unit direction.  Its
electronic norm is
\begin{equation}
\|\Psi_\R\|^2=\int |\Psi(\rr,\R)|^2\,d\rr.
\label{eq:Psirnorm}
\end{equation}
Where this norm is nonzero, define
\begin{equation}
\chi(\R)=\|\Psi_\R\|,
\qquad
\Phi_\R(\rr)=\frac{\Psi(\rr,\R)}{\|\Psi_\R\|}.
\label{eq:normdirection}
\end{equation}
Then
\begin{equation}
\Psi(\rr,\R)=\chi(\R)\Phi_\R(\rr),
\qquad
\int |\Phi_\R(\rr)|^2\,d\rr=1.
\label{eq:EFansatz}
\end{equation}
The restriction to points where the marginal norm is nonzero is substantive: zeros of the
marginal introduce a state-dependent regularity question, analyzed by Jecko, Sutcliffe, and
Woolley and clarified in their corrigendum~\cite{JeckoSutcliffeWoolley2015,JSW2018corrigendum}.
This is the exact factorization used in the present paper.  The factor $|\chi(\R)|^2$ is the
nuclear probability density obtained after the electronic coordinates are integrated out;
$\Phi_\R$ is the corresponding normalized electronic factor at that nuclear geometry.  In
Hunter's terminology these are marginal and conditional amplitudes
\cite{Hunter1974,Hunter1975}.  Modern molecular formulations were developed by Cederbaum and
by Abedi, Maitra, and Gross
\cite{Cederbaum2013,AbediMaitraGross2010,AbediMaitraGross2012}.

The relation to the Born--Huang coefficients follows immediately.  Orthonormality of the
electronic basis gives
\begin{equation}
|\chi(\R)|^2=\sum_k|\chi_k(\R)|^2.
\label{eq:coeffnorm}
\end{equation}
Define
\begin{equation}
C_k(\R)=\frac{\chi_k(\R)}{\chi(\R)}.
\label{eq:Ckdef}
\end{equation}
Then
\begin{equation}
\Phi_\R=\sum_k C_k(\R)\phi_k,
\qquad
\sum_k|C_k|^2=1,
\qquad
\chi_k=\chi C_k.
\label{eq:BH_EF_relation}
\end{equation}
The Born--Huang coefficients therefore describe the vector relative to a chosen electronic
basis, whereas exact factorization separates the same vector into its length and normalized
direction:
\begin{equation}
\begin{aligned}
\text{full BH:}&\quad \text{components of the electronic vector in a chosen basis},\\
\text{EF:}&\quad \text{length and normalized direction of the same vector}.
\end{aligned}
\label{eq:componentsdirection}
\end{equation}
No novelty is claimed for the relation itself; the purpose of this form is explanatory.  The
phase of the normalized direction will matter only when we differentiate it, and is therefore
deferred until Section~\ref{sec:metric}.

\subsection{What the nuclear kinetic energy does to the Born--Huang expansion}

We now ask what changes when Eq.~\eqref{eq:BHsum} is substituted into the molecular
Schr\"odinger equation.  Both $\phi_k(\rr;\R)$ and $\chi_k(\R)$ depend on the nuclear
coordinates, so the nuclear kinetic-energy operator differentiates both factors.  If a single
product $\phi_k\chi$ is considered, the product rule gives
\begin{equation}
\Tn(\phi_k\chi)
=\phi_k\Tn\chi
-\sum_\alpha\frac{1}{M_\alpha}(\nabla_\alpha\phi_k)\!\cdot\!(\nabla_\alpha\chi)
-\chi\sum_\alpha\frac{1}{2M_\alpha}\nabla_\alpha^2\phi_k.
\label{eq:productrule}
\end{equation}
The last two terms exist because the electronic eigenfunction changes with nuclear geometry.
Every one of them carries an inverse nuclear mass, but mass alone does not determine whether
they are small: their size also depends on how rapidly $\phi_k$ changes with $\R$.

For the complete expansion, substitute Eq.~\eqref{eq:BHsum} into
Eq.~\eqref{eq:fullSE}, multiply the result by $\phi_j^*(\rr;\R)$, and integrate over the
electronic coordinates.  This gives coupled differential equations for the nuclear
coefficients $\chi_j(\R)$.  Their coupling coefficients contain electronic integrals of
$\nabla_\alpha\phi_k$ and $\nabla_\alpha^2\phi_k$.  The full coupled set is still an exact
representation of the same molecular problem.

\subsection{The two one-surface approximations}

Only now do we restrict the exact Born--Huang expansion.  Retain one fixed-nuclei electronic
state $\phi_k$.  If all terms in Eq.~\eqref{eq:productrule} that differentiate $\phi_k$ are
omitted, denote the resulting one-surface energy by $\Ecal$.  The nuclear equation is
\begin{equation}
[\Tn+E_k(\R)]\chi=\Ecal\chi.
\label{eq:BOeq}
\end{equation}
We call this the \emph{Born--Oppenheimer approximation}.  The electronic eigenvalue
$E_k(\R)$ now appears as the potential-energy term in the nuclear Schr\"odinger equation.

A less severe one-state approximation keeps the part of the derivative contribution that
acts within the retained electronic state while still discarding coupling to the other
electronic states.  For the real electronic eigenfunctions used in our one-dimensional
benchmarks, normalization gives $\bk{\phi_k}{\nabla_\alpha\phi_k}=0$.  Let $W_k(\R)$ denote
the retained scalar derivative term:
\begin{equation}
W_k(\R)=\sum_\alpha\frac{1}{2M_\alpha}
\bk{\nabla_\alpha\phi_k}{\nabla_\alpha\phi_k}.
\label{eq:Wdef}
\end{equation}
Denote the energy of this corrected one-surface equation by $E^A$, where the superscript $A$ labels the corrected adiabatic one-surface energy.  The nuclear equation is
\begin{equation}
[\Tn+E_k(\R)+W_k(\R)]\chi=E^A\chi.
\label{eq:BHeq}
\end{equation}
We call this the \emph{single-surface Born--Huang} approximation
\cite{Kutzelnigg1997,HandyYamaguchiSchaefer1986}.  Equation~\eqref{eq:Wdef} is the form needed
for the real one-dimensional benchmark eigenfunctions below.  Section~\ref{sec:metric}
returns to the derivative of a normalized electronic state in the general complex case.

Terminology is not uniform.  Some current literature includes the diagonal derivative
correction in what it calls the Born--Oppenheimer approximation and calls the uncorrected
equation ``adiabatic BO''; the corrected potentials are then called Born--Huang surfaces
\cite{AgostiniCurchod2026,PrljEtAl2026}.  We use the convention above because the paper
compares the uncorrected potential-energy function $E_k$ and the corrected function $E_k+W_k$ separately.

\begin{table}[H]
\centering
\small
\renewcommand{\arraystretch}{1.18}
\begin{tabular}{@{}>{\raggedright\arraybackslash}p{0.18\textwidth}
                >{\raggedright\arraybackslash}p{0.25\textwidth}
                >{\raggedright\arraybackslash}p{0.22\textwidth}
                >{\raggedright\arraybackslash}p{0.25\textwidth}@{}}
\toprule
Description & Mathematical form & Exact or approximate? & What is retained or discarded \\
\midrule
Full molecular problem & $\Psi(\rr,\R)$ & exact & no electron--nuclear separation has been made \\
Full BH & $\Psi=\sum_k\chi_k\phi_k$ & exact with complete spectral resolution & all electronic components and the coupled nuclear equations are retained \\
BO & one retained $\phi_k$, potential $E_k$ & approximate & all terms produced by nuclear derivatives of $\phi_k$ are discarded \\
Single-surface BH & one retained $\phi_k$, potential $E_k+W_k$ & approximate & the within-state derivative term is retained; coupling to other electronic states is discarded \\
EF & $\Psi=\chi\Phi_\R$ & exact for the specified molecular state & the full molecular state is written as one nuclear amplitude and one normalized electronic factor \\
\bottomrule
\end{tabular}
\caption{The descriptions used in the paper, stated in terms of the actual mathematical
operations.}
\label{tab:three}
\end{table}

\section{The Fern\'andez model: one Hamiltonian, four calculations}
\label{sec:fernandez}

The first benchmark asks what can be learned when the exact molecular problem is simple
enough that the exact, BO, single-surface BH, and EF descriptions can all be written down
explicitly.  Let $x$ be a unit-mass coordinate playing the role of an electronic coordinate,
let $X$ be a nuclear-like coordinate with mass $M$, and let $\beta$ be their bilinear
coupling.  We consider
\begin{equation}
\hat H=-\frac12\partial_x^2-\frac{1}{2M}\partial_X^2
      +\frac12(x^2+X^2)+\beta xX,
\qquad \beta^2<1.
\label{eq:fern}
\end{equation}
Fern\'andez used this model to test his mass-ratio expansion through sixth order
\cite{Fernandez1994}.  McCoy's bilinear oscillator becomes Eq.~\eqref{eq:fern} after setting
$\hbar=m=\omega_r=1$, $\omega_R=M^{-1/2}$, and $V_1=\beta$ \cite{McCoy2001}.  For the
physical fast--slow interpretation we take $M>1$, although the ground-state proposition
below needs only $M>0$.

We use the following energy notation throughout this section:
\begin{equation}
\begin{array}{ll}
E_{nv} & \text{exact total molecular energy},\\[2pt]
\Ecal_{nv} & \text{BO approximate molecular energy},\\[2pt]
E^A_{nv} & \text{single-surface BH approximate molecular energy}.
\end{array}
\label{eq:fernenergynotation}
\end{equation}
The meaning of the indices $n$ and $v$ follows from the exact normal modes below.

\subsection{Exact molecular energies}

Set $\widetilde X=\sqrt M\,X$.  In the mass-weighted coordinates $(x,\widetilde X)$ the
kinetic energy has unit mass in both directions, and an orthogonal rotation diagonalizes the
quadratic potential.  Denote the two squared normal-mode frequencies by $k_+$ and $k_-$:
\begin{equation}
k_\pm=\frac{(M+1)\pm\sqrt{(M-1)^2+4M\beta^2}}{2M}.
\label{eq:kpm}
\end{equation}
Let $n,v=0,1,2,\ldots$ denote the quantum numbers of the $+$ and $-$ normal modes,
respectively.  These labels are exact and do not require a large-mass approximation.  The
exact molecular energies are
\begin{equation}
E_{nv}=\left(n+\frac12\right)\sqrt{k_+}
       +\left(v+\frac12\right)\sqrt{k_-}.
\label{eq:exactspectrum}
\end{equation}
For $M>1$, the $+$ mode has the higher frequency.  In the large-$M$ limit,
\begin{equation}
\sqrt{k_+}\longrightarrow1,
\qquad
\sqrt{k_-}\sim\sqrt{\frac{1-\beta^2}{M}},
\label{eq:fastslowlimit}
\end{equation}
so the $+$ and $-$ modes acquire the natural interpretation as fast and slow motion.
We use the same label $n$ below for the fixed-$X$ electronic oscillator state because, in
this limit, the $+$ normal mode approaches that electronic motion; this gives the natural
correspondence between exact and one-surface levels used in the comparison.  Appendix
\ref{app:fernandez} gives the normal-mode calculation.

\subsection{Clamped electronic states and the BO calculation}

At fixed $X$,
\begin{equation}
\frac12x^2+\beta xX+\frac12X^2
=\frac12(x+\beta X)^2+\frac12(1-\beta^2)X^2.
\label{eq:complsq}
\end{equation}
Let $\varphi_n$ denote the normalized $n$th eigenfunction of the unit-frequency harmonic
oscillator.  The fixed-$X$ electronic eigenfunction is
\begin{equation}
\phi_n(x;X)=\varphi_n(x+\beta X),
\label{eq:fernphi}
\end{equation}
and the corresponding fixed-$X$ electronic eigenvalue is
\begin{equation}
U_n(X)=n+\frac12+\frac12(1-\beta^2)X^2.
\label{eq:Un}
\end{equation}
When nuclear motion on this electronic state is calculated, $U_n(X)$ is the
potential-energy function in the BO nuclear equation.  The separation between any two
fixed-$X$ electronic eigenvalues is
\begin{equation}
U_m(X)-U_n(X)=m-n,
\label{eq:ferngap}
\end{equation}
which is independent of $X$.

The nuclear motion on every $U_n$ is harmonic with frequency
\begin{equation}
\omega=\sqrt{\frac{1-\beta^2}{M}}.
\label{eq:omegafern}
\end{equation}
Let $\chi_v^{(\omega)}(X)$ be the normalized $v$th eigenfunction of this nuclear oscillator.
The BO molecular wavefunction and energy are therefore
\begin{align}
\Psi^{\rm BO}_{nv}(x,X)
&=\phi_n(x;X)\chi_v^{(\omega)}(X),
\label{eq:fernBOwf}\\
\Ecal_{nv}
&=n+\frac12+\left(v+\frac12\right)\omega.
\label{eq:BOfernenergy}
\end{align}

\subsection{What changes when the diagonal derivative term is kept}

From Eq.~\eqref{eq:fernphi},
\begin{equation}
\partial_X\phi_n=\beta\,\varphi_n'(x+\beta X).
\label{eq:fernphideriv}
\end{equation}
The normalization $\langle\varphi_n|\varphi_n\rangle=1$ does not imply that the derivative
has unit norm.  For the unit-frequency oscillator,
\begin{equation}
\langle\varphi_n'|\varphi_n'\rangle
=\langle\varphi_n|\hat p^2|\varphi_n\rangle
=n+\frac12,
\label{eq:oscderivnorm}
\end{equation}
where the last equality follows from the harmonic-oscillator virial theorem; Appendix
\ref{app:fernandez} gives the equivalent ladder-operator calculation.  Equation
\eqref{eq:Wdef} then gives
\begin{equation}
W_n=\frac{\beta^2}{2M}\left(n+\frac12\right).
\label{eq:Wn}
\end{equation}
The important simplification is that $W_n$ is independent of $X$.

The BO nuclear equation is
\begin{equation}
[\Tn+U_n(X)]\chi_v^{(\omega)}=\Ecal_{nv}\chi_v^{(\omega)},
\label{eq:fernBOeqexplicit}
\end{equation}
whereas the single-surface BH equation is
\begin{equation}
[\Tn+U_n(X)+W_n]\chi^{A}_{nv}=E^A_{nv}\chi^{A}_{nv}.
\label{eq:fernBHeqexplicit}
\end{equation}
Adding the constant $W_n$ shifts every eigenvalue by $W_n$ without changing the
eigenfunction.  Hence
\begin{equation}
\chi^{A}_{nv}=\chi_v^{(\omega)},
\qquad
E^A_{nv}=\Ecal_{nv}+W_n.
\label{eq:BHfernenergy}
\end{equation}
For this particular model the BO and single-surface BH molecular wavefunctions are therefore
identical,
\begin{equation}
\Psi^{\rm BH}_{nv}(x,X)=\Psi^{\rm BO}_{nv}(x,X),
\label{eq:fernBHwf}
\end{equation}
even though their energies differ.

\subsection{The same exact state in full BH and EF form}

The exact molecular eigenstate with normal-mode labels $(n,v)$ can also be expanded in the
fixed-$X$ electronic basis,
\begin{equation}
\Psi_{nv}(x,X)=\sum_j\chi_j^{(nv)}(X)\phi_j(x;X).
\label{eq:fernfullBH}
\end{equation}
This is the full Born--Huang description of the exact state.  Applying the construction of
Section~\ref{sec:EFconstruction} to these coefficients gives
\begin{equation}
\chi_{nv}^{\rm EF}(X)
=\left[\sum_j|\chi_j^{(nv)}(X)|^2\right]^{1/2},
\qquad
C_j^{(nv)}(X)=\frac{\chi_j^{(nv)}(X)}{\chi_{nv}^{\rm EF}(X)},
\label{eq:fernBHtoEFcoeffs}
\end{equation}
and
\begin{equation}
\Phi_{nv,X}(x)=\sum_j C_j^{(nv)}(X)\phi_j(x;X),
\qquad
\Psi_{nv}(x,X)=\chi_{nv}^{\rm EF}(X)\Phi_{nv,X}(x).
\label{eq:fernEFgeneral}
\end{equation}
The model therefore gives an explicit realization of the distinction introduced in Section
\ref{sec:EFconstruction}: full BH supplies the components of the exact electronic vector,
whereas EF supplies its length and normalized direction.

\subsection{The ground-state EF potential}

For the ground state the exact normal-mode wavefunction is
\begin{equation}
\Psi_{00}(z_+,z_-)
=\frac{(\omega_+\omega_-)^{1/4}}{\sqrt\pi}
\exp\!\left[-\frac12\left(\omega_+z_+^2+\omega_-z_-^2\right)\right],
\qquad \omega_\pm=\sqrt{k_\pm}.
\label{eq:fernnormalground}
\end{equation}
Write the normal-mode rotation so that
\begin{equation}
\sqrt M\,X=\sin\gamma\,z_++\cos\gamma\,z_-.
\label{eq:fernrotationX}
\end{equation}
The two normal coordinates are independent Gaussians with
$\langle z_\pm^2\rangle=(2\omega_\pm)^{-1}$.  Therefore
\begin{equation}
\langle X^2\rangle
=\frac{1}{2M}\left(\frac{\sin^2\gamma}{\omega_+}
                   +\frac{\cos^2\gamma}{\omega_-}\right).
\label{eq:fernXvariance}
\end{equation}
A marginal of a Gaussian is Gaussian, so the exact EF nuclear amplitude can be written
\begin{equation}
\chi_{00}^{\rm EF}(X)=\left(\frac{a}{\pi}\right)^{1/4}e^{-aX^2/2},
\qquad
a=\frac{1}{2\langle X^2\rangle}.
\label{eq:fernEFamp}
\end{equation}
Equivalently,
\begin{equation}
\frac1a=\frac1M\left(\frac{\sin^2\gamma}{\omega_+}
                    +\frac{\cos^2\gamma}{\omega_-}\right).
\label{eq:fern_a}
\end{equation}
Thus Eq.~\eqref{eq:fern_a} converts the exactly known normal-mode covariance into the width of
the exact nuclear marginal.  Hunter had already applied the marginal--conditional construction
to coupled harmonic oscillators~\cite{Hunter1974}; the formulas below specialize that exact
factorization to Fern\'andez's parameterization and place it beside the BO and single-surface
BH descriptions.

Before deriving the general EF nuclear equation in Section~\ref{sec:needEF}, define the
state-dependent potential associated with this nodeless marginal by requiring
\begin{equation}
\left[-\frac{1}{2M}\frac{d^2}{dX^2}+\varepsilon_{00}(X)\right]
\chi_{00}^{\rm EF}(X)=E_{00}\chi_{00}^{\rm EF}(X).
\label{eq:fernEFnucleareq}
\end{equation}
Equivalently,
\begin{equation}
\varepsilon_{00}(X)=E_{00}+\frac{1}{2M}
\frac{(\chi_{00}^{\rm EF})''(X)}{\chi_{00}^{\rm EF}(X)}.
\label{eq:fernEFinversion}
\end{equation}
Section~\ref{sec:needEF} derives Eq.~\eqref{eq:fernEFnucleareq} from exact factorization and
thereby identifies this definition with the exact EF scalar potential for the real, nonnegative
nuclear factor fixed in Eq.~\eqref{eq:normdirection}.
For Eq.~\eqref{eq:fernEFamp},
\begin{equation}
(\chi_{00}^{\rm EF})''=(a^2X^2-a)\chi_{00}^{\rm EF}.
\label{eq:fernEFsecondderiv}
\end{equation}
Defining $\omega_{\rm EF}=a/M$ therefore gives
\begin{equation}
\boxed{
\varepsilon_{00}(X)
=E_{00}-\frac{\omega_{\rm EF}}2
 +\frac{M\omega_{\rm EF}^2}{2}X^2.}
\label{eq:fernEF}
\end{equation}
The exact EF nuclear equation for this molecular ground state is therefore itself a harmonic
oscillator equation, with a frequency determined by the exact nuclear marginal.

The same Gaussian also permits a direct comparison between the exact EF electronic direction
and the clamped electronic ground state.  In Hunter's terminology this is the conditional
amplitude associated with the marginal just constructed~\cite{Hunter1974,Hunter1975}.  Rewrite Eq.~\eqref{eq:fernnormalground} in the
coordinates $(x,X)$ and let
\begin{equation}
A_{xx}=\omega_+\cos^2\gamma+\omega_-\sin^2\gamma
\label{eq:fernAxx}
\end{equation}
denote the coefficient multiplying $x^2$ in the Gaussian exponent.  At each fixed $X$, divide
the molecular Gaussian by the marginal amplitude in Eq.~\eqref{eq:fernEFamp} and normalize
the resulting electronic function.  Differentiating that normalized conditional function
with respect to $X$ gives
\begin{equation}
\bk{\partial_X\Phi_{00,X}}{\partial_X\Phi_{00,X}}
=\frac{M(\omega_+-\omega_-)^2\sin^2\gamma\cos^2\gamma}{2A_{xx}},
\label{eq:fernEFmetricclosed}
\end{equation}
which is independent of $X$.  This is the same type of derivative norm that appeared in
the single-surface BH correction; Section~\ref{sec:metric} will show why, for a real
normalized electronic state, it is also the quantum-metric coefficient.  At $X=0$ the same
conditional Gaussian gives
\begin{equation}
\bk{\Phi_{00,0}}{\Hel(0)\Phi_{00,0}}
=\frac{A_{xx}}{4}+\frac{1}{4A_{xx}},
\label{eq:fernEFHelclosed}
\end{equation}
whereas the clamped electronic ground-state eigenvalue is $U_0(0)=1/2$.  For
$M=10$ and $\beta=0.5$, Eqs.~\eqref{eq:fernEFmetricclosed}--\eqref{eq:fernEFHelclosed}
give $\langle\partial_X\Phi_{00,X}|\partial_X\Phi_{00,X}\rangle=0.0764$, compared with
$\langle\partial_X\phi_0|\partial_X\phi_0\rangle=\beta^2/2=0.125$, and
$\langle\Phi_{00,0}|\Hel(0)|\Phi_{00,0}\rangle=0.500015$.  Thus the exact EF direction is
not the clamped electronic eigenstate at finite nuclear mass.  The ratio of the two
derivative norms approaches one as $M\to\infty$.

\begin{table}[H]
\centering
\small
\renewcommand{\arraystretch}{1.15}
\begin{tabular}{@{}p{0.23\textwidth}p{0.30\textwidth}p{0.30\textwidth}@{}}
\toprule
Calculation & Potential used in the nuclear equation & Energy \\
\midrule
BO & $U_n(X)$ & $\Ecal_{nv}$ \\
Single-surface BH & $U_n(X)+W_n$ & $E^A_{nv}$ \\
EF ground state & state-dependent $\varepsilon_{00}(X)$ & exact $E_{00}$ \\
\bottomrule
\end{tabular}
\caption{Three nuclear equations obtained from the same Fern\'andez Hamiltonian.  The full
BH expansion is not a fourth approximation; when it is not truncated, it reproduces the
exact molecular spectrum.}
\label{tab:ferncompare}
\end{table}

\subsection{Ground-state energy ordering and excited-state counterexamples}

Fern\'andez explicitly stated, without restricting the statement to the vibrational ground state, that the Born--Oppenheimer
and corrected adiabatic energies bracket the exact energy of the system, credited the result to
Brattsev, and then verified the strict inequality for the ground state of this bilinear oscillator
~\cite{Fernandez1994,Brattsev1965}.  The excited-state rows below show why the ground-state
qualifier is necessary for this model.  His warning about
terminology is important: Brattsev called the lower quantity ``adiabatic,'' whereas Fern\'andez
and the present paper call it the Born--Oppenheimer energy; Brattsev's ``variational'' energy is
the corrected adiabatic, or single-surface BH, energy used here.  Epstein supplied a corrected
proof of the lower-bound statement~\cite{Epstein1966}, and McCoy later rederived that inequality
and checked it for the bilinear oscillator~\cite{McCoy2001}.  The bracket itself is therefore
not new.  What the closed-form Fern\'andez formulas allow us to add is the pair of exact
difference identities below, which prove strictness for every $M>0$ and $0<|\beta|<1$, rather
than only within the $M>1$ convergence range of the mass-ratio expansion.

\begin{proposition}[Strict two-sided ground-state energy ordering in the Fern\'andez model]
\label{prop:brattsev}
For $M>0$ and $0<|\beta|<1$,
\begin{equation}
\Ecal_{00}<E_{00}<E^A_{00}.
\label{eq:brattsevorder}
\end{equation}
\end{proposition}

\begin{proof}
From Eqs.~\eqref{eq:kpm}--\eqref{eq:BHfernenergy},
\[
E_{00}=\frac12(\sqrt{k_+}+\sqrt{k_-}),\qquad
\Ecal_{00}=\frac12(1+\omega),\qquad
E^A_{00}=\Ecal_{00}+\frac{\beta^2}{4M}.
\]
The trace and determinant of the normal-mode stiffness matrix give
$k_++k_-=1+M^{-1}$ and $k_+k_-=(1-\beta^2)/M=\omega^2$.  Hence
\begin{align}
(2E_{00})^2-(2\Ecal_{00})^2&=\frac{\beta^2}{M}>0,\\
(2E^A_{00})^2-(2E_{00})^2
&=\frac{\beta^2\omega}{M}+\frac{\beta^4}{4M^2}>0.
\end{align}
All three energies are positive, proving Eq.~\eqref{eq:brattsevorder}.

The upper inequality also follows from the Rayleigh--Ritz variational principle.  For a
self-adjoint Hamiltonian with ground energy $E_0$,
\begin{equation}
E_0=\inf_{\|\Psi\|=1}\langle\Psi|\hat H|\Psi\rangle.
\label{eq:rayleighritz}
\end{equation}
Restricting the minimization to normalized products
$\Psi(x,X)=\phi_0(x;X)\chi(X)$ can only increase the minimum.  The restricted expectation
value is exactly the single-surface BH functional, so its minimum is $E^A_{00}$ and
$E_{00}\le E^A_{00}$.  This is the standard variational principle; an accessible account is
given in Ref.~\cite{AtkinsFriedman2011}.
\end{proof}

This is a ground-state statement.  At $M=10$ and $\beta=0.1$ the exact formulas give
\begin{center}
\small
\begin{tabular}{c|ccc|c}
$(n,v)$ & $\Ecal_{nv}$ & $E_{nv}$ & $E^A_{nv}$ & ordering\\ \hline
$(0,0)$ & $0.657321$ & $0.657511$ & $0.657571$ & $\Ecal<E<E^A$\\
$(0,1)$ & $0.971964$ & $0.971980$ & $0.972214$ & $\Ecal<E<E^A$\\
$(0,2)$ & $1.286607$ & $1.286448$ & $1.286857$ & $E<\Ecal<E^A$\\
$(2,0)$ & $2.657321$ & $2.658621$ & $2.658571$ & $\Ecal<E^A<E$\\
\end{tabular}
\end{center}
The $(0,2)$ row shows that the positive diagonal correction can move an excited energy
farther from the exact value.  The $(2,0)$ row shows that even the corrected one-surface
energy can lie below the exact excited energy.  Ordinary Rayleigh--Ritz minimization is a
ground-state statement; excited-state variational results require additional orthogonality or
minimax conditions.  Section~\ref{sec:UQ} will show explicitly where the dependence on $v$
enters the Fern\'andez expansion.

Fern\'andez expands the exact energies in the large-mass variable $\eta=M^{-1/4}$.  A Taylor
series about a point converges up to the nearest singularity of the analytically continued
function in the complex plane \cite{BenderOrszag1999}.  In the present model the singularities
nearest $\eta=0$ come from the zeros of the square-root discriminant,
\begin{equation}
M_\pm=(1-2\beta^2)\pm2i\beta\sqrt{1-\beta^2},
\qquad |M_\pm|=1.
\label{eq:branchpoints}
\end{equation}
Since $\eta=M^{-1/4}$, these branch points map to $|\eta|=1$.  The expansion about $\eta=0$
therefore has radius of convergence one and converges for positive real $M>1$
\cite{Fernandez1994}.  The stability condition is equally direct: at $|\beta|=1$ the slow
force constant vanishes, and for $|\beta|>1$ the quadratic potential is indefinite.
Figure~\ref{fig:fern} summarizes these properties.

\begin{figure}[p]
\centering
\begin{minipage}{0.48\textwidth}
\centering
\includegraphics[width=\linewidth]{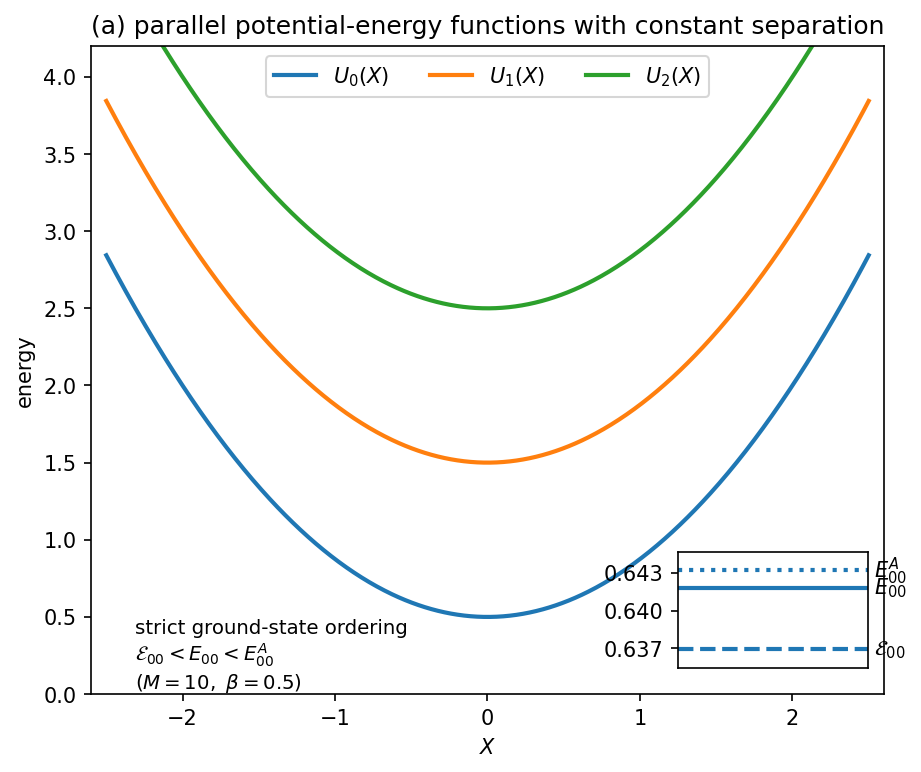}
\end{minipage}\hfill
\begin{minipage}{0.48\textwidth}
\centering
\includegraphics[width=\linewidth]{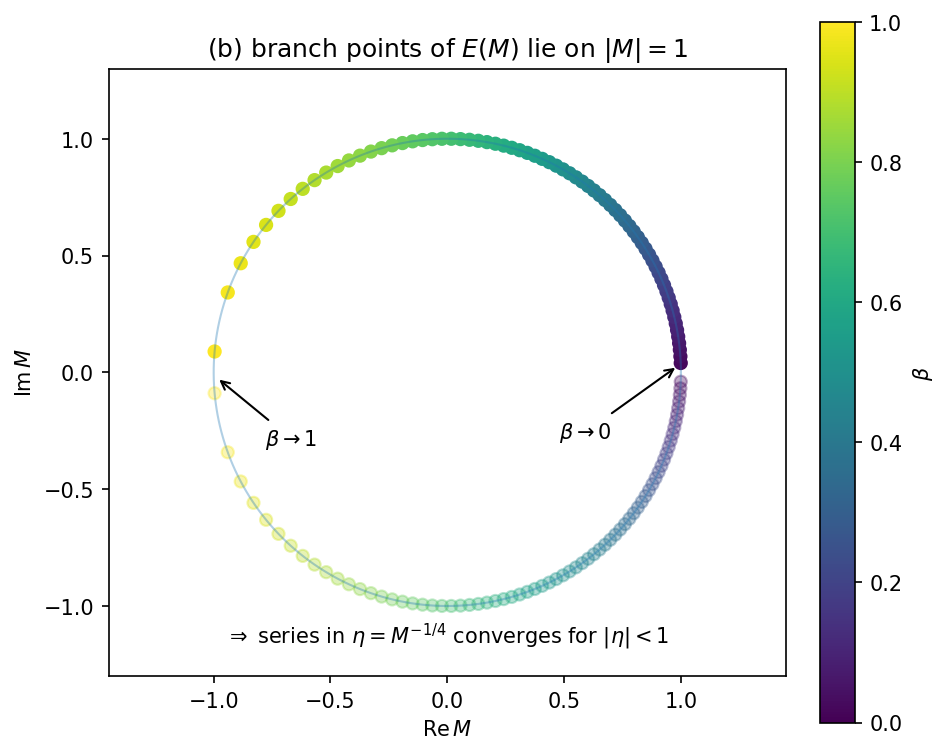}
\end{minipage}

\vspace{0.7em}
\includegraphics[width=0.52\textwidth]{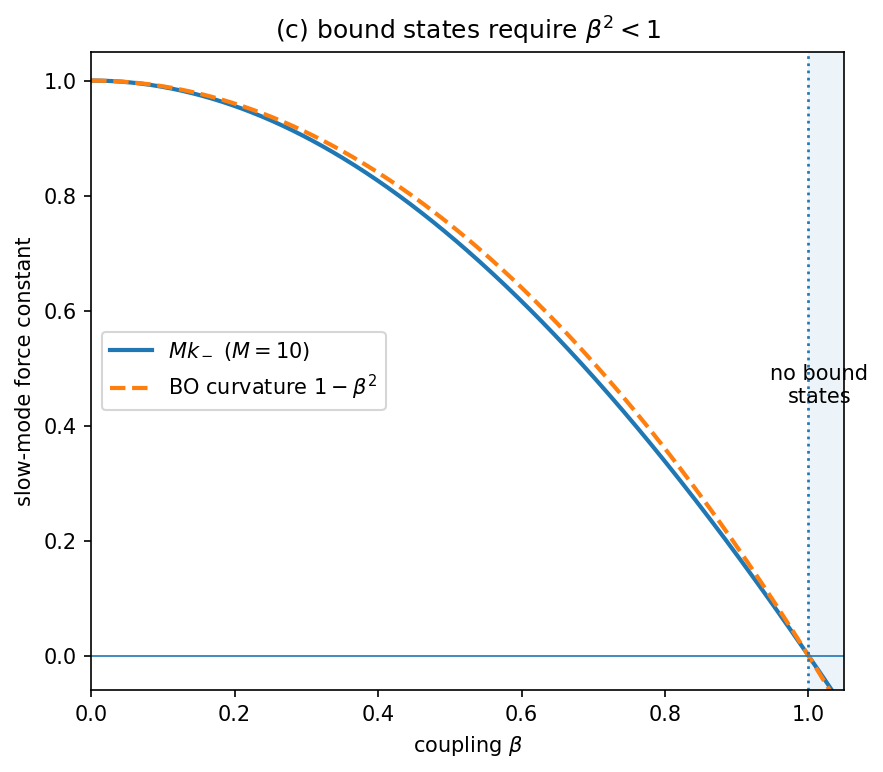}
\caption{Exact structural features of the Fern\'andez model.  \textbf{(a)} The BO
potential-energy functions are parallel parabolas with constant separation; the inset
compares exact, BO, and single-surface BH ground energies at $M=10$, $\beta=0.5$.
\textbf{(b)} The complex branch points controlling the mass-ratio expansion lie on
$|M|=1$.  \textbf{(c)} The slow-mode force constant vanishes at the stability threshold
$|\beta|=1$.}
\label{fig:fern}
\end{figure}
\FloatBarrier

The model therefore serves several purposes: it supplies exact molecular energies,
it makes the BO and BH nuclear equations explicit, and it gives an exact EF nuclear
potential for the ground state.  It also has a nonzero derivative of the electronic state.
What it cannot test is the effect of an electronic energy separation that changes with
nuclear position.  We first examine the derivative itself, and then introduce a second model
for that remaining question.

\section{From the derivative of a normalized electronic state to quantum geometry}
\label{sec:metric}

The derivative of an electronic eigenstate has already appeared in Sections
\ref{sec:hierarchy} and \ref{sec:fernandez}.  We now ask what part of that derivative describes
an actual change of the normalized electronic state and what part can be changed simply by
changing its phase.

\subsection{Separate phase change from change of electronic direction}

For one nuclear coordinate $R$, let $\ket{\phi(R)}$ be normalized and let $\theta(R)$ be any
real function.  The vectors
\begin{equation}
\ket{\widetilde\phi(R)}=e^{i\theta(R)}\ket{\phi(R)}
\label{eq:gaugetransform}
\end{equation}
represent the same physical electronic state.  To see which part of the derivative depends
on this arbitrary phase convention, differentiate Eq.~\eqref{eq:gaugetransform} with respect
to $R$:
\begin{equation}
\ket{\partial_R\widetilde\phi}
=e^{i\theta}\left(\ket{\partial_R\phi}+i\theta'\ket{\phi}\right).
\label{eq:phasederivative}
\end{equation}
Changing the phase therefore adds a term parallel to $\ket\phi$.

Let
\begin{equation}
P=\ket\phi\bra\phi.
\label{eq:projectorP}
\end{equation}
For any electronic vector $|v\rangle$, $P|v\rangle$ is its component parallel to $|\phi\rangle$,
whereas $(1-P)|v\rangle$ is the perpendicular component.  Applying this decomposition to the
derivative gives
\begin{equation}
\ket{\partial_R\phi}
=P\ket{\partial_R\phi}+(1-P)\ket{\partial_R\phi}.
\label{eq:parperp}
\end{equation}
Because $|\widetilde\phi\rangle\langle\widetilde\phi|=P$, Eqs.~\eqref{eq:phasederivative}
and \eqref{eq:parperp} give
\begin{align}
P\ket{\partial_R\widetilde\phi}
&=e^{i\theta}\left[P\ket{\partial_R\phi}+i\theta'\ket\phi\right],
\label{eq:paralleltransform}\\
(1-P)\ket{\partial_R\widetilde\phi}
&=e^{i\theta}(1-P)\ket{\partial_R\phi}.
\label{eq:perptransform}
\end{align}
The parallel part can therefore be changed by the arbitrary function $\theta'(R)$.  The
perpendicular part changes only by the same overall phase as the state itself, so its squared
norm is unchanged.

The coefficient of the parallel part is obtained directly from the projector,
\begin{equation}
P\ket{\partial_R\phi}=\ket\phi\bk{\phi}{\partial_R\phi}.
\label{eq:parallelcoeff}
\end{equation}
Differentiating the normalization condition $\langle\phi|\phi\rangle=1$ gives
\begin{equation}
\bk{\partial_R\phi}{\phi}+\bk{\phi}{\partial_R\phi}=0.
\label{eq:normderivative}
\end{equation}
The two terms are complex conjugates, so $\bk{\phi}{\partial_R\phi}$ is purely imaginary.
It is therefore convenient to define the real quantity
\begin{equation}
A(R)=-i\bk{\phi}{\partial_R\phi}.
\label{eq:berryconnection}
\end{equation}
Equation~\eqref{eq:phasederivative} then gives
\begin{equation}
\widetilde A(R)=A(R)+\theta'(R).
\label{eq:Atransform}
\end{equation}
The quantity $A$ is usually called the Berry connection \cite{Berry1984}.  The name is less
important here than the calculation: $A$ records the part of the derivative that depends on
how the phase of the normalized vector has been chosen as $R$ changes.

The perpendicular part is independent of that choice in the sense of
Eq.~\eqref{eq:perptransform}.  Define its squared norm by
\begin{equation}
g(R)=\bk{\partial_R\phi}{(1-P)\partial_R\phi}.
\label{eq:metric1d}
\end{equation}
We next show why this quantity measures distance between nearby physical electronic states.

\subsection{Nearby normalized states and the quantum metric}

Consider the normalized electronic states at $R$ and $R+dR$.  Their inner product
\begin{equation}
\bk{\phi(R)}{\phi(R+dR)}
\label{eq:overlapdef}
\end{equation}
is commonly called their overlap.  Its magnitude is one if the two normalized vectors differ
only by a phase.  We therefore ask how its magnitude changes when the physical electronic
state itself changes.

Expand
\begin{equation}
\ket{\phi(R+dR)}
=\ket\phi+dR\ket{\phi'}+\frac12dR^2\ket{\phi''}+O(dR^3),
\label{eq:phitaylor}
\end{equation}
where primes denote $R$ derivatives.  Differentiating
$\langle\phi|\phi\rangle=1$ once and twice gives the relations needed to simplify the square
of the overlap.  A short calculation yields
\begin{equation}
|\bk{\phi(R)}{\phi(R+dR)}|^2
=1-\left[\bk{\phi'}{\phi'}-|\bk{\phi}{\phi'}|^2\right]dR^2+O(dR^3).
\label{eq:overlapexpanded}
\end{equation}
The bracket is exactly Eq.~\eqref{eq:metric1d}, so
\begin{equation}
|\bk{\phi(R)}{\phi(R+dR)}|^2
=1-g(R)dR^2+O(dR^3).
\label{eq:overlap1d}
\end{equation}
Thus
\begin{equation}
ds^2=g(R)dR^2
\label{eq:FSlineelement}
\end{equation}
is the infinitesimal squared distance along the family of normalized electronic states.  The
coefficient $g(R)$ is the one-dimensional Fubini--Study, or quantum, metric coefficient
\cite{ProvostVallee1980}.  Recent molecular work has used electronic-state overlaps to
extract the same geometric information over nuclear configuration space
\cite{XieLiuGu2025}.

\subsection{What determines how rapidly a clamped electronic eigenstate changes?}

We have found that $g(R)$ measures how rapidly the normalized electronic state changes with
nuclear position.  We now ask what determines that rate of change when the state is a
clamped electronic eigenstate.  Start from
\begin{equation}
\Hel(R)\ket{\phi_k(R)}=E_k(R)\ket{\phi_k(R)}
\label{eq:clampedket}
\end{equation}
and differentiate with respect to $R$:
\begin{equation}
(\partial_R\Hel)\ket{\phi_k}
+\Hel\ket{\partial_R\phi_k}
=E_k'\ket{\phi_k}+E_k\ket{\partial_R\phi_k}.
\label{eq:diffclamped}
\end{equation}
Multiply on the left by $\bra{\phi_j}$ with $j\ne k$.  Using orthonormality and
$\bra{\phi_j}\Hel=E_j\bra{\phi_j}$ gives
\begin{equation}
(E_k-E_j)\bk{\phi_j}{\partial_R\phi_k}
=\bk{\phi_j}{(\partial_R\Hel)\phi_k}.
\label{eq:HFstep}
\end{equation}
Define the derivative coupling
\begin{equation}
d_{jk}(R)=\bk{\phi_j}{\partial_R\phi_k}.
\label{eq:djkdef}
\end{equation}
Then
\begin{equation}
d_{jk}(R)=
\frac{\bk{\phi_j}{(\partial_R\Hel)\phi_k}}{E_k(R)-E_j(R)}.
\label{eq:HF}
\end{equation}
This is the off-diagonal Hellmann--Feynman relation
\cite{Hellmann1937,Feynman1939}.  It shows explicitly that, when the numerator is nonzero, a
smaller separation between electronic eigenvalues can make the change of the eigenstate with
$R$ larger.

The metric coefficient follows from completeness.  Let
$P_k=|\phi_k\rangle\langle\phi_k|$.  Inserting
$1=\sum_j|\phi_j\rangle\langle\phi_j|$ into Eq.~\eqref{eq:metric1d} gives
\begin{equation}
g_k(R)=\sum_{j\ne k}|d_{jk}(R)|^2.
\label{eq:metricsum}
\end{equation}
Combining Eqs.~\eqref{eq:HF} and \eqref{eq:metricsum} gives
\begin{equation}
g_k(R)=\sum_{j\ne k}
\frac{|\bk{\phi_j}{(\partial_R\Hel)\phi_k}|^2}
{[E_j(R)-E_k(R)]^2}.
\label{eq:metricspectral1d}
\end{equation}
This is the form we will use to understand why the electronic energy separation matters.

\subsection{The diagonal BH correction is the metric term}

We now return to the operation that produced the diagonal BH correction.  In one nuclear
coordinate,
\begin{equation}
-\frac{1}{2M}\partial_R^2(\phi\chi)
=-\frac{1}{2M}\left(\phi\chi''+2\phi'\chi'+\phi''\chi\right).
\label{eq:kineticproduct1d}
\end{equation}
Multiply by $\langle\phi|$ and integrate over the electronic coordinates:
\begin{equation}
-\frac{1}{2M}\left[
\chi''+2\bk{\phi}{\phi'}\chi'+\bk{\phi}{\phi''}\chi
\right].
\label{eq:projectedkinetic1d}
\end{equation}
Using $A=-i\langle\phi|\phi'\rangle$ and differentiating
$\langle\phi|\phi'\rangle$ gives
\begin{equation}
\bk{\phi}{\phi''}=iA'-\bk{\phi'}{\phi'}.
\label{eq:phisecondidentity}
\end{equation}
Together with
\begin{equation}
g=\bk{\phi'}{\phi'}-|\bk{\phi}{\phi'}|^2
=\bk{\phi'}{\phi'}-A^2,
\label{eq:gexpanded}
\end{equation}
we obtain
\begin{equation}
\left\langle\phi\left|-
\frac{1}{2M}\partial_R^2\right|\phi\chi\right\rangle
=
\left[\frac{1}{2M}(-i\partial_R+A)^2+\frac{g}{2M}\right]\chi.
\label{eq:projectedkineticgeometry}
\end{equation}
The scalar term produced by the change of the electronic state is therefore
\begin{equation}
\boxed{W_k(R)=\frac{g_k(R)}{2M}.}
\label{eq:WisG}
\end{equation}
For the real electronic eigenfunctions used in Section~\ref{sec:fernandez},
$\langle\phi_k|\partial_R\phi_k\rangle=0$ and hence $A=0$, so Eq.~\eqref{eq:WisG}
reduces exactly to the derivative norm used there.

For several Cartesian nuclear coordinates $R^\mu$, define
\begin{equation}
g^k_{\mu\nu}(\R)
=\operatorname{Re}\bk{\partial_\mu\phi_k}
{(1-P_k)\partial_\nu\phi_k}.
\label{eq:metricmulti}
\end{equation}
The nuclear kinetic-energy operator contains only the diagonal second derivatives in these
coordinates, so its contraction with the metric retains the diagonal components:
\begin{equation}
W_k(\R)=\sum_\mu\frac{g^k_{\mu\mu}(\R)}{2M_\mu}.
\label{eq:Wmetricmulti}
\end{equation}
This is why the one-dimensional result in Eq.~\eqref{eq:WisG} generalizes to a single sum,
even though the metric itself has two coordinate indices.  The relation between this metric
and the diagonal Born--Oppenheimer/Born--Huang correction has been emphasized in recent
molecular quantum-geometry work
\cite{KronenbergerSteinmetzAppenzellerPausch2025,XieLiuGu2025}.

For the Fern\'andez oscillator, Appendix~\ref{app:fernandez} gives
\begin{equation}
g_n=\beta^2\left(n+\frac12\right),
\label{eq:ferngmetric}
\end{equation}
which is independent of $X$.  The electronic state changes with $X$, but the metric
coefficient does not.  Section~\ref{sec:vibronic} introduces a second model in which the
electronic energy separation changes with nuclear position.

\section{What nuclear equation follows from exact factorization?}
\label{sec:needEF}

Section~\ref{sec:EFconstruction} showed that the exact molecular wavefunction can be written
as
\begin{equation}
\Psi(\rr,R)=\chi(R)\Phi_R(\rr),
\qquad \langle\Phi_R|\Phi_R\rangle=1.
\label{eq:EFrecalled}
\end{equation}
Section~\ref{sec:metric} then showed how the derivative of any normalized electronic vector
can be separated into a phase-dependent parallel part and a perpendicular part that measures
change of the physical electronic state.  We now ask a different question: \emph{what equation
does the nuclear factor $\chi(R)$ satisfy when the exact factorization is substituted into the
exact molecular Schr\"odinger equation?}  The answer identifies the exact nuclear potential
supplied by EF and makes its relation to the single-surface BH derivative term explicit.

Equation~\eqref{eq:normdirection} fixed the convenient representative
$\chi(R)=\|\Psi_R\|$, real and nonnegative.  More generally, the product in
Eq.~\eqref{eq:EFrecalled} allows an $R$-dependent phase to be redistributed between its two
factors.  For any real $\theta(R)$,
\begin{equation}
\widetilde\Phi_R=e^{i\theta(R)}\Phi_R,
\qquad
\widetilde\chi=e^{-i\theta(R)}\chi
\label{eq:EFgaugeearly}
\end{equation}
gives exactly the same molecular wavefunction.  This freedom matters because the nuclear
kinetic-energy operator differentiates both factors: moving an $R$-dependent phase from one
factor to the other changes their derivatives even though their product is unchanged.

Section~\ref{sec:metric} defined, for any normalized electronic vector, the real coefficient
$A=-i\langle\phi|\partial_R\phi\rangle$.  Applying the same construction to the normalized EF
electronic factor gives
\begin{equation}
A(R)=-i\bk{\Phi_R}{\partial_R\Phi_R}.
\label{eq:EFconnectionearly}
\end{equation}
Under Eq.~\eqref{eq:EFgaugeearly}, the calculation already carried out in Section
\ref{sec:metric} gives $\widetilde A=A+\theta'$.  The combination that will occur in the
nuclear equation transforms as
\begin{equation}
(-i\partial_R+\widetilde A)\widetilde\chi
=e^{-i\theta}(-i\partial_R+A)\chi.
\label{eq:covariantcombo}
\end{equation}
Thus the separate phase choices change, but the differential expression that enters the
nuclear equation changes only by the same overall phase as $\chi$.

To obtain that nuclear equation, substitute Eq.~\eqref{eq:EFrecalled} into the exact
molecular Schr\"odinger equation.  We first carry out the calculation for one nuclear
coordinate because every term can then be seen explicitly:
\begin{equation}
\left[-\frac{1}{2M}\partial_R^2+\Hel(R)\right]\chi(R)\Phi_R
=E\chi(R)\Phi_R.
\label{eq:EFsubstitutionstart}
\end{equation}
The product rule gives
\begin{equation}
-\frac{1}{2M}\left[\Phi_R\chi''+2\Phi_R'\chi'+\chi\Phi_R''\right]
+\chi\Hel\Phi_R
=E\chi\Phi_R.
\label{eq:EFproductexpanded}
\end{equation}
Multiply by $\langle\Phi_R|$ and integrate over the electronic coordinates.  Using
$\langle\Phi_R|\Phi_R\rangle=1$ gives
\begin{equation}
-\frac{1}{2M}\left[
\chi''+2\bk{\Phi_R}{\Phi_R'}\chi'
+\bk{\Phi_R}{\Phi_R''}\chi
\right]
+\bk{\Phi_R}{\Hel\Phi_R}\chi
=E\chi.
\label{eq:EFprojectedraw}
\end{equation}
From Eq.~\eqref{eq:EFconnectionearly},
$\langle\Phi_R|\Phi_R'\rangle=iA$.  Differentiating this inner product gives
\begin{equation}
\bk{\Phi_R}{\Phi_R''}=iA'-\bk{\Phi_R'}{\Phi_R'}.
\label{eq:EFsecondidentity}
\end{equation}
Substituting into Eq.~\eqref{eq:EFprojectedraw} and completing the square yields
\begin{equation}
\left[
\frac{1}{2M}(-i\partial_R+A)^2
+\bk{\Phi_R}{\Hel\Phi_R}
+\frac{1}{2M}\left(
\bk{\Phi_R'}{\Phi_R'}-A^2
\right)
\right]\chi=E\chi.
\label{eq:EFoneDintermediate}
\end{equation}
The last bracket is the same perpendicular-derivative quantity constructed in Section
\ref{sec:metric}.  Define
\begin{equation}
g^{\Phi}(R)
=\bk{\Phi_R'}{(1-|\Phi_R\rangle\langle\Phi_R|)\Phi_R'}
=\bk{\Phi_R'}{\Phi_R'}-A^2.
\label{eq:EFmetric1d}
\end{equation}
The exact nuclear equation is therefore
\begin{equation}
\left[
\frac{1}{2M}(-i\partial_R+A)^2+\varepsilon_\Psi(R)
\right]\chi=E\chi,
\label{eq:EFeq1d}
\end{equation}
with
\begin{equation}
\boxed{
\varepsilon_\Psi(R)
=\bk{\Phi_R}{\Hel(R)\Phi_R}
+\frac{g^{\Phi}(R)}{2M}.}
\label{eq:EFgeometry1d}
\end{equation}

For several Cartesian nuclear coordinates the same calculation gives
\begin{equation}
\left[
\sum_\mu\frac{1}{2M_\mu}(-i\nabla_\mu+A_\mu)^2
+\varepsilon_\Psi(\R)
\right]\chi=E\chi,
\label{eq:EFeq}
\end{equation}
where
\begin{equation}
A_\mu(\R)=-i\bk{\Phi_\R}{\nabla_\mu\Phi_\R},
\qquad
\varepsilon_\Psi(\R)=
\bk{\Phi_\R}{\Hel\Phi_\R}
+\sum_\mu\frac{g^\Phi_{\mu\mu}(\R)}{2M_\mu}.
\label{eq:EFgeometry}
\end{equation}

We can now compare exact factorization with single-surface BH using quantities already
derived.  In one dimension,
\begin{equation}
\begin{array}{ll}
\text{single-surface BH:}&
U_k(R)+\dfrac{g_k(R)}{2M},\\[8pt]
\text{EF:}&
\langle\Phi_R|\Hel(R)|\Phi_R\rangle+\dfrac{g^\Phi(R)}{2M}.
\end{array}
\label{eq:BHvsEFgeometry}
\end{equation}
Both contain the perpendicular derivative of a normalized electronic vector, but the vector
is different.  Single-surface BH prescribes the clamped electronic eigenstate $\phi_k(R)$
before the nuclear equation is solved.  EF uses the normalized electronic direction
$\Phi_R$ extracted from the particular exact molecular state.  Moreover,
$\langle\Phi_R|\Hel|\Phi_R\rangle$ is generally not a clamped electronic eigenvalue because
$\Phi_R$ need not itself be one of the eigenstates $\phi_k$.

The exact scalar potential is therefore state dependent.  Two exact molecular eigenstates of
the same Hamiltonian generally give different $\chi$, different $\Phi_\R$, and different
$\varepsilon_\Psi$.  Exact factorization does not replace all BO potential-energy functions
by one universal exact function; it gives an exact nuclear equation for the molecular state
being factorized.  Projected-derivative geometry in exact factorization has been developed by
Requist, Tandetzky, and Gross \cite{RequistTandetzkyGross2016}.  In the benchmark calculations
below we first calculate the molecular eigenstate and then construct its nuclear amplitude and
normalized electronic factor; we do not use the coupled EF equations as a practical
replacement for the molecular calculation.

\section{A two-state model with position-dependent electronic energy separation}
\label{sec:vibronic}

The Fern\'andez model has given us an exact example in which the electronic eigenstate changes
as the nuclear coordinate changes.  Section~\ref{sec:metric} showed that the size of this
change is measured by the quantum metric and that the single-surface BH correction is the
metric divided by twice the nuclear mass.  The same calculation also showed that changes of
an electronic eigenstate involve the energy separations from the other electronic states.

There is, however, one feature that the Fern\'andez model cannot test.  Its fixed-$X$
electronic eigenvalues satisfy
\begin{equation}
U_m(X)-U_n(X)=m-n,
\end{equation}
so every electronic energy separation is independent of $X$.  We therefore introduce a
second model with one purpose: to make the electronic energy separation depend on nuclear
position while keeping every electronic calculation elementary and explicit.  Two-state linear
vibronic-coupling Hamiltonians of this type have a long history in non-Born--Oppenheimer
molecular dynamics~\cite{KoppelDomckeCederbaum1984}; here the model is used only as a
transparent benchmark for the gap dependence of the derivative geometry.

\subsection{Constructing the two-state electronic problem}

Use the simplest electronic space containing two fixed orthonormal components,
\begin{equation}
|1\rangle=\begin{pmatrix}1\\0\end{pmatrix},
\qquad
|2\rangle=\begin{pmatrix}0\\1\end{pmatrix}.
\label{eq:fixedbasisvectors}
\end{equation}
These basis vectors do not depend on the nuclear coordinate $R$.  Before coupling the two
components, assign them the $R$-dependent electronic energies $+\kappa R$ and $-\kappa R$,
where $\kappa>0$.  Their electronic matrix would be
\begin{equation}
h_0(R)=\begin{pmatrix}\kappa R&0\\0&-\kappa R\end{pmatrix}.
\label{eq:h0matrix}
\end{equation}
We now allow the two components to interact by adding a constant real coupling
$\lambda>0$ between them:
\begin{equation}
h(R)=\begin{pmatrix}
\kappa R&\lambda\\
\lambda&-\kappa R
\end{pmatrix}.
\label{eq:hmatrix}
\end{equation}
Here $\kappa$ sets how quickly the two uncoupled energies separate as $|R|$ increases, while
$\lambda$ sets the coupling between the fixed electronic components.

The electronic eigenvalues follow from
\begin{equation}
\det[h(R)-eI]=e^2-\kappa^2R^2-\lambda^2=0.
\label{eq:twostatedeterminant}
\end{equation}
Define
\begin{equation}
\rho(R)=\sqrt{\kappa^2R^2+\lambda^2}.
\label{eq:rhodef}
\end{equation}
Then
\begin{equation}
e_\pm(R)=\pm\rho(R),
\qquad
\Delta(R)=e_+(R)-e_-(R)=2\rho(R).
\label{eq:vibgap}
\end{equation}
The separation is smallest at $R=0$,
\begin{equation}
\Delta_{\min}=\Delta(0)=2\lambda,
\label{eq:vibgapmin}
\end{equation}
and increases with $|R|$.  Thus $\lambda$ lets us vary the smallest electronic energy
separation while retaining explicit eigenvalues and eigenvectors.

The electronic matrix alone is enough to study its eigenvectors, but Section~\ref{sec:EFvib}
will also compare nuclear ground states.  We therefore add the same harmonic confining
potential to both electronic components.  Because this term is proportional to the identity,
it does not change either the electronic eigenvectors or their separation.  The full
molecular Hamiltonian is
\begin{equation}
\hat H=-\frac{1}{2M}\frac{d^2}{dR^2}I
+\frac12M\Omega^2R^2 I+h(R).
\label{eq:vibronic}
\end{equation}
The two potential-energy functions appearing in the corresponding one-surface nuclear
equations are therefore
\begin{equation}
U_\pm(R)=\frac12M\Omega^2R^2\pm\rho(R).
\label{eq:Usurfaces}
\end{equation}
Their separation is still $\Delta(R)$ in Eq.~\eqref{eq:vibgap}.

\subsection{How rapidly do the electronic eigenvectors change?}

Section~\ref{sec:metric} showed that the BH correction depends on how the clamped electronic
eigenvector changes with $R$.  We therefore need the eigenvectors of $h(R)$.  Since
Eq.~\eqref{eq:hmatrix} is a real symmetric $2\times2$ matrix, its normalized eigenvectors are
obtained by a rotation of the fixed basis.  It is convenient to describe that rotation by an
angle $\vartheta(R)$ defined by
\begin{equation}
\cos\vartheta=\frac{\kappa R}{\rho},
\qquad
\sin\vartheta=\frac{\lambda}{\rho}.
\label{eq:mixangle}
\end{equation}
A convenient pair of real normalized eigenvectors is
\begin{equation}
\phi_+(R)=
\begin{pmatrix}\cos(\vartheta/2)\\ \sin(\vartheta/2)\end{pmatrix},
\qquad
\phi_-(R)=
\begin{pmatrix}\sin(\vartheta/2)\\ -\cos(\vartheta/2)\end{pmatrix}.
\label{eq:explicitvibevecs}
\end{equation}

We introduced the model to determine how the changing electronic energy separation affects
this change of electronic direction.  Section~\ref{sec:metric} tells us to calculate
\begin{equation}
d_{+-}(R)=\bk{\phi_+}{\partial_R\phi_-}.
\label{eq:dcoupquestion}
\end{equation}
Differentiating Eq.~\eqref{eq:explicitvibevecs} gives
\begin{equation}
d_{+-}(R)=\frac12\vartheta'(R).
\label{eq:dcouptheta}
\end{equation}
From Eq.~\eqref{eq:mixangle},
\begin{equation}
\vartheta'(R)=-\frac{\kappa\lambda}{\kappa^2R^2+\lambda^2},
\label{eq:thetaderiv}
\end{equation}
so
\begin{equation}
d_{+-}(R)=-\frac{\kappa\lambda}{2(\kappa^2R^2+\lambda^2)}.
\label{eq:dcoupvib}
\end{equation}
There is only one electronic state orthogonal to $\phi_-$, so the metric coefficient and the
single-surface BH correction are
\begin{equation}
\boxed{
g_{RR}(R)=|d_{+-}(R)|^2
=\frac{\kappa^2\lambda^2}{4(\kappa^2R^2+\lambda^2)^2},
\qquad
W_\pm(R)=\frac{g_{RR}(R)}{2M}.}
\label{eq:gW}
\end{equation}
At $R=0$,
\begin{equation}
g_{RR}(0)=\frac{\kappa^2}{4\lambda^2}.
\label{eq:gcenter}
\end{equation}
The characteristic width of this peak is of order $\lambda/\kappa$.  Thus decreasing the
minimum electronic energy separation $2\lambda$ makes the electronic eigenvector change more
rapidly over a smaller interval of $R$.

The same calculation also shows where the coupling appears in two different electronic
representations.  In the fixed basis of Eq.~\eqref{eq:fixedbasisvectors}, the basis vectors
do not depend on $R$, so nuclear differentiation of the basis gives zero; the electronic
matrix instead contains the off-diagonal entry $\lambda$.  In the eigenvector basis of
Eq.~\eqref{eq:explicitvibevecs}, the electronic matrix is diagonal, but the basis vectors
depend on $R$ and nuclear differentiation produces $d_{+-}$.  Only after this calculation do
we attach the usual labels: the fixed representation is commonly called diabatic and the
eigenvector representation adiabatic \cite{PrljEtAl2026}.

\begin{table}[H]
\centering
\small
\renewcommand{\arraystretch}{1.15}
\begin{tabular}{@{}>{\raggedright\arraybackslash}p{0.22\textwidth}>{\raggedright\arraybackslash}p{0.25\textwidth}>{\raggedright\arraybackslash}p{0.24\textwidth}>{\raggedright\arraybackslash}p{0.20\textwidth}@{}}
\toprule
Electronic vectors & Electronic Hamiltonian & Nuclear derivative of basis vectors & Coupling appears as \\
\midrule
fixed components (diabatic) & nondiagonal & zero & off-diagonal electronic matrix element $\lambda$ \\
clamped eigenvectors (adiabatic) & diagonal & nonzero & derivative coupling $d_{+-}$ \\
\bottomrule
\end{tabular}
\caption{Two descriptions of the same two-state Hamiltonian.  The mathematical difference
is stated before the usual representation labels are introduced.}
\label{tab:diabaticadiabatic}
\end{table}

\subsection{Does a smaller minimum separation produce a larger total electronic change?}

The metric peak becomes higher and narrower as $\lambda$ decreases.  Does that mean that the
electronic state undergoes a larger total change, or that the same change is concentrated
into a smaller interval of $R$?  Equation~\eqref{eq:FSlineelement} gives the accumulated
distance along the family of normalized electronic states:
\begin{equation}
L(R)=\int_{-\infty}^{R}\sqrt{g_{RR}(R')}\,dR'
=\frac12\left[
\arctan\!\left(\frac{\kappa R}{\lambda}\right)+\frac{\pi}{2}
\right].
\label{eq:FSlength}
\end{equation}
Therefore
\begin{equation}
\boxed{L(+\infty)=\frac{\pi}{2},}
\label{eq:FSpi2}
\end{equation}
independent of $\kappa$ and $\lambda$.  The same conclusion is visible directly from the
mixing angle: as $R$ runs from $-\infty$ to $+\infty$, $\vartheta$ changes by $\pi$ for every
nonzero $\lambda$.  Decreasing $\lambda$ therefore does not increase the total change of the
electronic direction; it compresses that fixed change into a smaller region of nuclear
coordinate.

\begin{figure}[H]
\centering
\includegraphics[width=0.70\textwidth]{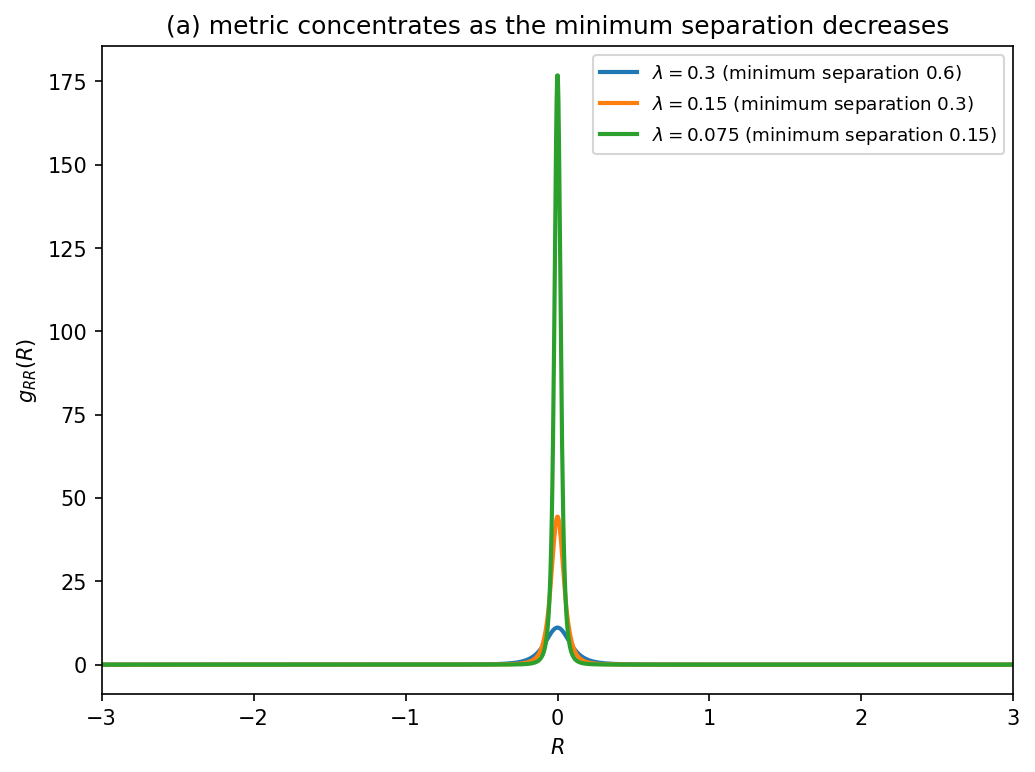}\par\vspace{0.7em}
\includegraphics[width=0.70\textwidth]{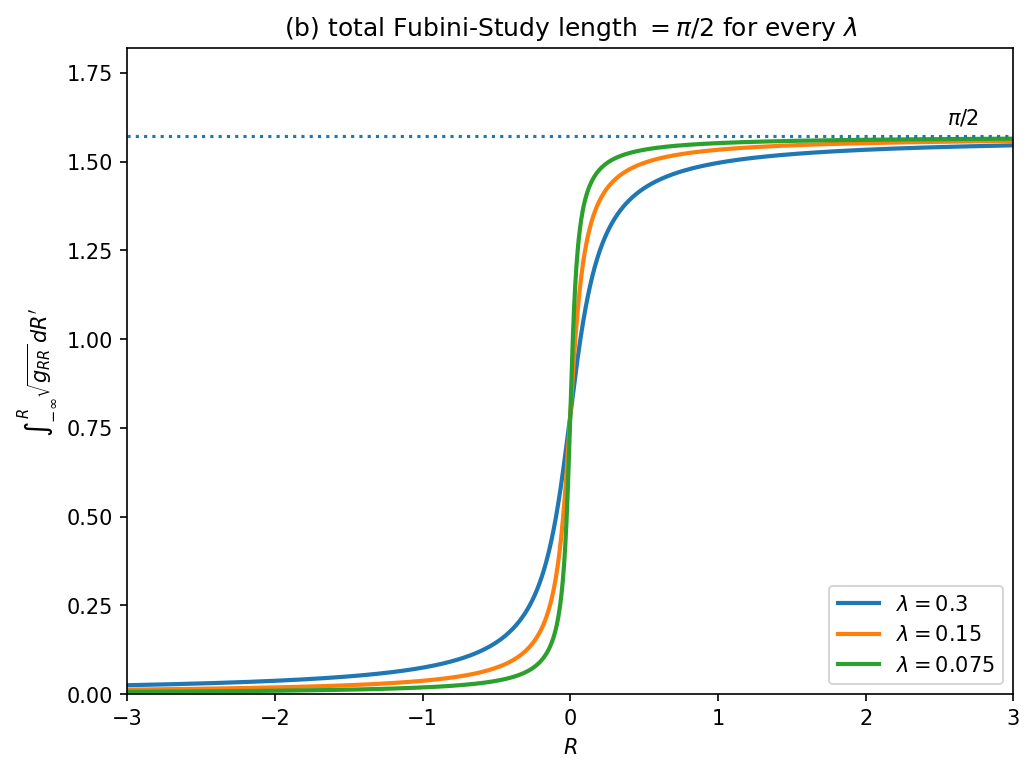}
\caption{Quantum geometry of the position-dependent-gap benchmark ($\kappa=2$).
\textbf{(a)} The metric $g_{RR}(R)$ becomes higher and narrower as the minimum electronic
energy separation $2\lambda$ decreases.  \textbf{(b)} The accumulated Fubini--Study length
approaches the same limit $\pi/2$ for every $\lambda$.  The changing energy separation
changes where the electronic direction changes most rapidly, not the total distance between
the limiting directions.}
\label{fig:geom}
\end{figure}
\FloatBarrier

\section{What do the three nuclear potentials predict in the position-dependent-gap model?}
\label{sec:EFvib}

Section~\ref{sec:vibronic} isolated the feature absent from the Fern\'andez model: the
electronic energy separation changes with nuclear position.  We found that the metric
coefficient, and hence the single-surface BH correction $W_-(R)$, now also varies with $R$.
This changes the nature of the BO--BH comparison.  In the Fern\'andez model $W_n$ is a
constant, so BO and single-surface BH have identical nuclear eigenfunctions and differ only
in their energies.  Here $W_-(R)$ changes the potential-energy function itself and can
therefore change the nuclear wavefunction.

We now calculate the ground state in three ways: the BO approximation, the single-surface BH
approximation, and the exact two-component molecular problem.  From the exact molecular
state we then construct the EF nuclear amplitude and potential.  This allows a direct
comparison of the three nuclear potentials, their ground-state energies, and the nuclear
wavefunctions they produce.

\subsection{The BO, single-surface BH, and exact molecular equations}

On the lower clamped electronic state $\phi_-(R)$, the BO nuclear equation is
\begin{equation}
\left[-\frac{1}{2M}\frac{d^2}{dR^2}+U_-(R)\right]
\chi_{\rm BO}(R)=\Ecal_0\chi_{\rm BO}(R),
\label{eq:vibBOeq}
\end{equation}
and the approximate BO molecular wavefunction is
\begin{equation}
\Psi_{\rm BO}(R)=\chi_{\rm BO}(R)\phi_-(R).
\label{eq:vibBOwf}
\end{equation}
The single-surface BH equation retains the position-dependent correction derived in Section
\ref{sec:vibronic}:
\begin{equation}
\left[-\frac{1}{2M}\frac{d^2}{dR^2}+U_-(R)+W_-(R)\right]
\chi_{\rm BH}(R)=E_0^A\chi_{\rm BH}(R),
\label{eq:vibBHeq}
\end{equation}
with approximate molecular wavefunction
\begin{equation}
\Psi_{\rm BH}(R)=\chi_{\rm BH}(R)\phi_-(R).
\label{eq:vibBHwf}
\end{equation}
Both one-surface approximations prescribe the same lower clamped electronic eigenvector; they
differ in the potential used to determine its nuclear coefficient.

The exact molecular problem is solved in the fixed two-component basis of Section
\ref{sec:vibronic}.  Write the exact ground state as
\begin{equation}
\Psi_0(R)=
\begin{pmatrix}\psi_1(R)\\ \psi_2(R)\end{pmatrix}
\label{eq:twocomp}
\end{equation}
with exact energy $E_0$.  Substituting this vector into Eq.~\eqref{eq:vibronic} gives the two
coupled nuclear equations
\begin{align}
\left[-\frac{1}{2M}\frac{d^2}{dR^2}
+\frac12M\Omega^2R^2+\kappa R\right]\psi_1
+\lambda\psi_2
&=E_0\psi_1,
\label{eq:vibexactcomp1}\\
\lambda\psi_1+
\left[-\frac{1}{2M}\frac{d^2}{dR^2}
+\frac12M\Omega^2R^2-\kappa R\right]\psi_2
&=E_0\psi_2.
\label{eq:vibexactcomp2}
\end{align}
These equations make explicit what is meant here by solving the full two-component molecular
problem.

\subsection{Constructing the exact-factorization amplitude and potential}

Once Eqs.~\eqref{eq:vibexactcomp1}--\eqref{eq:vibexactcomp2} have been solved, no additional
molecular approximation is needed to construct the EF factors.  At each $R$, the numbers
$\psi_1(R)$ and $\psi_2(R)$ are simply the components of the exact electronic vector.  We
write that same vector as its length times a normalized direction:
\begin{equation}
\chi_{\rm EF}(R)=\sqrt{\psi_1(R)^2+\psi_2(R)^2},
\qquad
\Phi_R=\frac{1}{\chi_{\rm EF}(R)}
\begin{pmatrix}\psi_1(R)\\ \psi_2(R)\end{pmatrix}.
\label{eq:vibmargcond}
\end{equation}
For the ground state $\chi_{\rm EF}$ can be chosen positive.

Section~\ref{sec:needEF} derived the exact EF nuclear equation.  For the present real
stationary ground state both factors can be chosen real, so
\begin{equation}
A_R=-i\bk{\Phi_R}{\partial_R\Phi_R}=0.
\label{eq:vibAzero}
\end{equation}
The exact nuclear equation therefore reduces to
\begin{equation}
\left[-\frac{1}{2M}\frac{d^2}{dR^2}+\varepsilon_{\rm EF}(R)\right]
\chi_{\rm EF}(R)=E_0\chi_{\rm EF}(R).
\label{eq:vibEFnucleareq}
\end{equation}
Solving this equation for the potential gives
\begin{equation}
\varepsilon_{\rm EF}(R)=E_0+\frac{1}{2M}
\frac{\chi_{\rm EF}''(R)}{\chi_{\rm EF}(R)}.
\label{eq:eps}
\end{equation}
Thus the exact molecular eigenstate determines an exact nuclear amplitude, and Eq.
\eqref{eq:eps} reconstructs the exact potential for the nuclear equation satisfied by that
amplitude.

Section~\ref{sec:needEF} also gave a second expression for the same potential in terms of the
normalized electronic direction.  Let
\begin{equation}
H_{\rm el}(R)=\frac12M\Omega^2R^2 I+h(R).
\label{eq:vibHeldef}
\end{equation}
Then
\begin{equation}
\varepsilon_{\rm EF}(R)
=\bk{\Phi_R}{H_{\rm el}(R)\Phi_R}
+\frac{1}{2M}\bk{\partial_R\Phi_R}{\partial_R\Phi_R}.
\label{eq:EFstructure}
\end{equation}
Because Eq.~\eqref{eq:vibAzero} gives $A_R=0$, the projected derivative term derived in Section~\ref{sec:needEF} reduces here to the derivative norm shown in Eq.~\eqref{eq:EFstructure}.  The single-surface
BH potential is
\begin{equation}
U_-(R)+W_-(R)
=U_-(R)+\frac{1}{2M}\bk{\partial_R\phi_-}{\partial_R\phi_-}.
\label{eq:BHstructurecompare}
\end{equation}
The resemblance is now concrete, but the two expressions differ in two ways.  First,
$\phi_-(R)$ is the prescribed lower clamped eigenstate, whereas $\Phi_R$ is the exact
normalized electronic direction extracted from the molecular ground state.  Second,
$\langle\Phi_R|H_{\rm el}|\Phi_R\rangle$ need not equal $U_-(R)$ because $\Phi_R$ need not be
a clamped electronic eigenvector.

\subsection{Numerical ground-state comparison}

Unlike the Fern\'andez Hamiltonian, the coupled nuclear equations
\eqref{eq:vibexactcomp1}--\eqref{eq:vibexactcomp2} do not provide the simple closed-form
solution needed for this comparison, so we calculate their ground state numerically.  We use
$R\in[-6,6]$ with Dirichlet boundaries and a second-order finite-difference approximation to
the nuclear second derivative.  Calculations on uniform grids with $N=1600$, $3200$, and
$6400$ are combined by Richardson extrapolation, using the expected $N^{-2}$ leading
discretization error.  The parameters are
\begin{equation}
M=4,\qquad \Omega=0.5,\qquad \kappa=2,\qquad \lambda=0.3.
\label{eq:vibparams}
\end{equation}
For this model we write
\begin{equation}
\begin{array}{ll}
\Ecal_0 & \text{BO ground-state energy},\\[2pt]
E_0 & \text{exact molecular ground-state energy},\\[2pt]
E_0^A & \text{single-surface BH ground-state energy}.
\end{array}
\label{eq:vibenergynotation}
\end{equation}
The extrapolated values are
\begin{equation}
\Ecal_0=-1.762524245,
\qquad
E_0=-1.762173980,
\qquad
E_0^A=-1.762161668.
\label{eq:vibenergies}
\end{equation}
For these parameter values,
\begin{equation}
\Ecal_0<E_0<E_0^A.
\label{eq:vibordering}
\end{equation}
The non-strict bracket $\Ecal_0\le E_0\le E_0^A$ is not merely numerical: the lower
inequality is the Brattsev--Epstein ground-state bound and the upper inequality follows from the
same Rayleigh--Ritz restriction used in Section~\ref{sec:fernandez}
~\cite{Brattsev1965,Epstein1966}.  What is numerical here is the strictness for this parameter
set and the magnitudes of the two differences.

The three potential-energy functions differ visibly near $R=0$, where the two electronic
eigenvalues are closest.  We therefore ask whether this local difference produces a
comparably large difference in the nuclear ground states.  For normalized real nuclear
amplitudes, the inner product is one when the two functions are identical and decreases as
they differ.  The calculated overlaps are
\begin{equation}
\bk{\chi_{\rm EF}}{\chi_{\rm BO}}=0.999984005,
\qquad
\bk{\chi_{\rm EF}}{\chi_{\rm BH}}=0.999999973.
\label{eq:viboverlaps}
\end{equation}
Thus both one-surface nuclear amplitudes are very close to the exact EF amplitude, and the BH
result is substantially closer.  A visible local change in a potential need not produce an
equally visible change in a particular nuclear state: the effect depends on where the nuclear
wavefunction has appreciable amplitude and on the size of the potential change relative to
the other energy scales in the nuclear equation.

\begin{table}[H]
\centering
\small
\renewcommand{\arraystretch}{1.15}
\begin{tabular}{@{}>{\raggedright\arraybackslash}p{0.17\textwidth}
                >{\raggedright\arraybackslash}p{0.18\textwidth}
                >{\raggedright\arraybackslash}p{0.16\textwidth}
                >{\raggedright\arraybackslash}p{0.14\textwidth}
                >{\raggedright\arraybackslash}p{0.17\textwidth}@{}}
\toprule
Calculation & Electronic state used & Nuclear potential & Ground energy & Overlap with exact EF amplitude \\
\midrule
BO & lower clamped eigenstate $\phi_-$ & $U_-$ & $\Ecal_0$ & $0.999984005$ \\
Single-surface BH & lower clamped eigenstate $\phi_-$ & $U_-+W_-$ & $E_0^A$ & $0.999999973$ \\
EF & exact direction $\Phi_R$ & $\varepsilon_{\rm EF}$ & $E_0$ & $1$ \\
\bottomrule
\end{tabular}
\caption{Ground-state comparison for Eq.~\eqref{eq:vibparams}.  The energies use
second-order finite differences followed by Richardson extrapolation over
$N=1600$--$6400$.}
\label{tab:EFcompare}
\end{table}

\begin{figure}[p]
\centering
\includegraphics[width=0.74\textwidth]{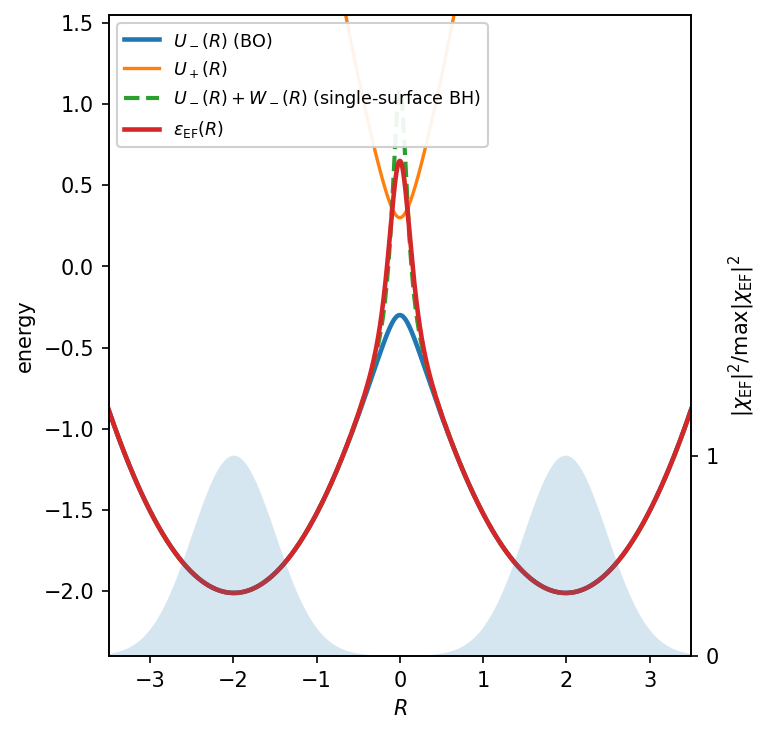}
\caption{Ground-state potential-energy functions in the position-dependent-gap model for the
parameters in Eq.~\eqref{eq:vibparams}: the BO functions $U_\pm$, the single-surface BH
potential $U_-+W_-$, the exact EF potential $\varepsilon_{\rm EF}$, and the exact nuclear
marginal density.  The largest local differences occur near $R=0$, where the electronic
energy separation is smallest.}
\label{fig:EFpot}
\end{figure}

\begin{figure}[p]
\centering
\includegraphics[width=0.88\textwidth]{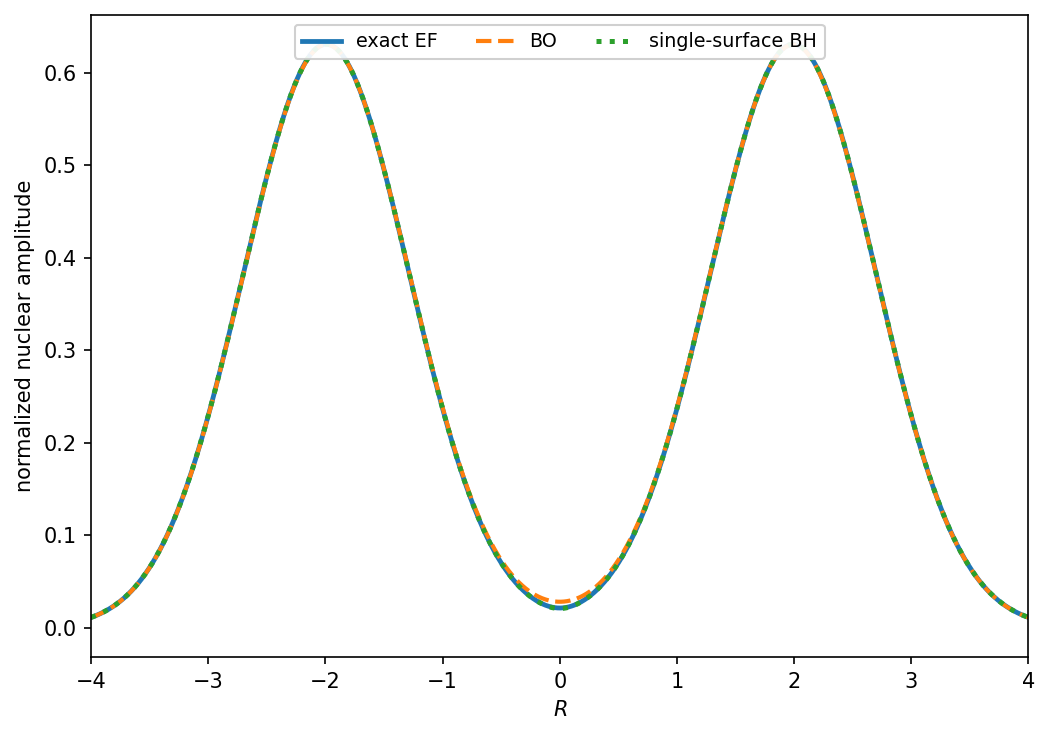}
\caption{The exact EF nuclear amplitude compared with the BO and single-surface BH nuclear
ground states for Eq.~\eqref{eq:vibparams}.  The effective potentials differ most near
$R=0$, but the three normalized nuclear amplitudes are nearly coincident.}
\label{fig:EFwf}
\end{figure}
\FloatBarrier

The comparison has therefore added a point that the local metric alone cannot make.  The
metric identifies where the electronic state changes rapidly, but the effect of that region
on an energy or nuclear wavefunction depends on where the nuclear state itself is large.
Section~\ref{sec:UQ} now combines this observation with the exact Fern\'andez expansion to
identify the quantities that control one-surface error.

\section{What controls the error of a one-surface description?}
\label{sec:UQ}

The preceding calculations have shown that there is no single statement such as
``single-surface BH is more accurate than BO'' that applies to every molecular state.  In the
Fern\'andez model the BH correction provides the complementary upper bound for the ground
state, but the ordering changes for excited states.  In the two-state model the BH correction
is strongly concentrated where the electronic energy separation is smallest, yet the
single-surface BH nuclear ground state remains extremely close to the exact EF amplitude.
We now ask whether these observations have a common explanation.

Because the Fern\'andez energies are known exactly, we can subtract the BO energy from the
exact molecular energy and expand the difference for large nuclear mass.  This tells us which
terms the one-surface description contains and which terms it misses.  Following Fern\'andez,
let
\begin{equation}
\eta=M^{-1/4}
\label{eq:etadefUQ}
\end{equation}
be a convenient bookkeeping variable.  On the lowest electronic surface the exact spectrum
gives
\begin{equation}
E_{0v}-\Ecal_{0v}
=\frac{\beta^2}{2M}
\left[\frac12-\left(v+\frac12\right)
\sqrt{\frac{1-\beta^2}{M}}\right]+O(\eta^8).
\label{eq:fernbudget}
\end{equation}
Appendix~\ref{app:budget} obtains this expression directly from the exact normal-mode
frequencies.  The two displayed contributions have different origins, and the quantities
needed to identify them have already appeared earlier in the paper.

For the lowest Fern\'andez electronic state, Eq.~\eqref{eq:ferngmetric} gives
\begin{equation}
g_{n=0}=\frac{\beta^2}{2}.
\label{eq:g0explicit}
\end{equation}
For brevity we write this as $g_0$ below.  The first term in Eq.~\eqref{eq:fernbudget} is
therefore
\begin{equation}
\frac{g_0}{2M}=\frac{\beta^2}{4M},
\label{eq:firstbudgetterm}
\end{equation}
which is exactly the positive diagonal BH correction.  Thus the leading positive
contribution to the exact BO energy error is the term already identified in Section
\ref{sec:metric} with the change of the retained electronic eigenstate.

The second contribution contains the nuclear-state dependence.  Let $\chi_v^{(\omega)}$ be
the BO nuclear oscillator eigenfunction introduced in Section~\ref{sec:fernandez}.  For a
harmonic oscillator the average kinetic and potential energies are equal, so
\begin{equation}
\bk{\chi_v^{(\omega)}}{\Tn\chi_v^{(\omega)}}
=\frac12\left(v+\frac12\right)
\sqrt{\frac{1-\beta^2}{M}}.
\label{eq:Tnfern}
\end{equation}
The same derivative couplings that form the metric occur in the next term of the exact
large-mass expansion, but with one additional division by the electronic energy separation.
This inverse gap follows directly from Fern\'andez's general result, not from the special
unit-gap oscillator.  His Eq.~(47), in the present notation, contains
\[
K=\sum_{j>0}
\frac{|\langle\phi_j|(\partial V/\partial u)|\phi_0\rangle|^2}
{(U_j-U_0)^3}.
\]
Using the off-diagonal Hellmann--Feynman relation in Eq.~\eqref{eq:HF}, the numerator contributes
two powers of $U_j-U_0$, so that
\[
K=\sum_{j>0}\frac{|d_{j0}|^2}{U_j-U_0}.
\]
Fern\'andez's Eq.~(48) supplies the accompanying nuclear-state factor and the negative sign of
the sixth-order nonadiabatic contribution~\cite{Fernandez1994}.  For the lowest electronic
state the gap-weighted combination in our notation is therefore
\begin{equation}
\sum_{j\ne0}\frac{|d_{j0}|^2}{U_j-U_0}.
\label{eq:gapweightedsum}
\end{equation}
In the Fern\'andez model only the $j=1$ derivative coupling contributes when $n=0$, and
$U_1-U_0=1$, so
\begin{equation}
\sum_{j\ne0}\frac{|d_{j0}|^2}{U_j-U_0}=\frac{\beta^2}{2}=g_0.
\label{eq:gapweightedsumfern}
\end{equation}
The numerical equality with $g_0$ is special to this unit-gap model; the two expressions
contain different information in general.

Using Eqs.~\eqref{eq:Tnfern} and \eqref{eq:gapweightedsum}, the exact expansion becomes
\begin{equation}
\boxed{
E_{0v}-\Ecal_{0v}
=\frac{g_0}{2M}
-\frac{2}{M}
\bk{\chi_v^{(\omega)}}{\Tn\chi_v^{(\omega)}}
\sum_{j\ne0}\frac{|d_{j0}|^2}{U_j-U_0}
+O(\eta^8).}
\label{eq:geometricbudget}
\end{equation}
This is Fern\'andez's sixth-order result written in quantities whose roles have already been
identified \cite{Fernandez1994}.  The first correction depends on how rapidly the retained
electronic state changes and carries the expected factor $1/M$.  The next correction depends
on the nuclear kinetic energy and on the same derivative couplings, but it also contains one
additional division by the electronic energy separation.

Equation~\eqref{eq:geometricbudget} explains the nuclear-excitation counterexample $(0,2)$
in Section~\ref{sec:fernandez}.  The positive diagonal term is independent of $v$, whereas the
magnitude of the negative omitted-state term grows with $v+1/2$.  Increasing nuclear excitation
can therefore make the omitted term larger than the retained diagonal correction, so keeping the
positive BH correction alone need not improve the excited-state energy.  It does not explain the
$(2,0)$ reversal: the budget displayed here is for $n=0$, and through sixth order Fern\'andez's
Eq.~(57) keeps $E_{nv}-E^A_{nv}$ negative for every $n$.  The electronic-excitation reversal
first appears at eighth order, through the term
$(n+\tfrac12)\beta^2(4-5\beta^2)\eta^8/8$ in his Eq.~(57)~\cite{Fernandez1994}.
The two counterexamples therefore have different perturbative origins.  McCoy had already
shown in the same bilinear family that a large fast--slow frequency separation is not by itself
sufficient; the coupling must also be small relative to that separation~\cite{McCoy2001}.

The Fern\'andez model cannot show how the omitted-state term changes when the electronic
energy separation itself varies with nuclear position.  The two-state model supplies that
missing comparison.  There is only one omitted electronic state, so the corresponding local
factor is
\begin{equation}
\frac{|d_{+-}(R)|^2}{U_+(R)-U_-(R)}
=\frac{g_{RR}(R)}{\Delta(R)}.
\label{eq:Gvib}
\end{equation}
At $R=0$, Eqs.~\eqref{eq:vibgapmin} and \eqref{eq:gcenter} give
\begin{equation}
g_{RR}(0)\propto\lambda^{-2},
\qquad
\frac{g_{RR}(0)}{\Delta(0)}\propto\lambda^{-3}.
\label{eq:gapscaling}
\end{equation}
Thus, as the minimum electronic energy separation decreases, the quantity entering the first
omitted-state correction grows faster than the diagonal BH metric term itself.  This is a
local statement; Section~\ref{sec:EFvib} has already shown why the effect on a particular
energy or nuclear wavefunction also depends on where that nuclear state is appreciable.

\begin{table}[H]
\centering
\small
\renewcommand{\arraystretch}{1.18}
\begin{tabular}{@{}>{\raggedright\arraybackslash}p{0.28\textwidth}
                >{\raggedright\arraybackslash}p{0.29\textwidth}
                >{\raggedright\arraybackslash}p{0.32\textwidth}@{}}
\toprule
Question & Fern\'andez model & Two-state model \\
\midrule
Electronic energy separation & constant & depends on $R$ \\
Electronic-state change & constant metric coefficient & concentrated where the two electronic energies are closest \\
Quantity varied in the comparison & nuclear excitation $v$ & minimum electronic separation $2\lambda$ \\
What becomes important & nuclear-state dependence & additional inverse-gap dependence \\
Main lesson & BH need not improve excited-state energies & smaller electronic separation magnifies omitted-state effects \\
\bottomrule
\end{tabular}
\caption{The two benchmarks isolate different parts of the one-surface error mechanism.}
\label{tab:twobenchmarks}
\end{table}

\begin{figure}[H]
\centering
\includegraphics[width=0.60\textwidth]{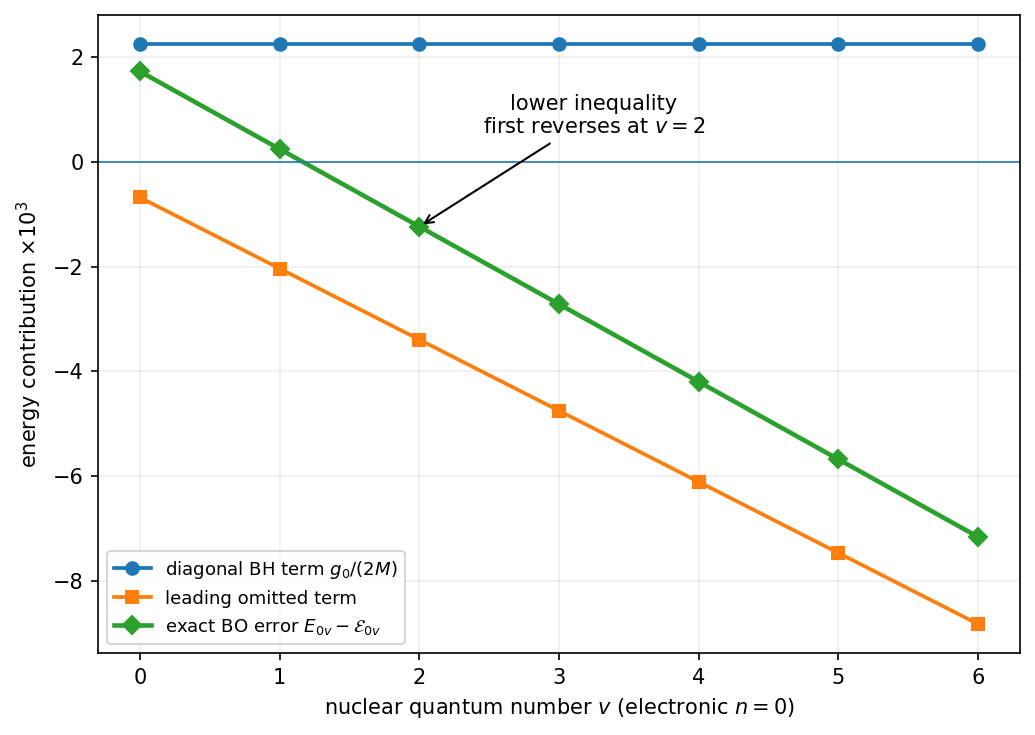}
\caption{Fern\'andez model ($M=10$, $\beta=0.3$, electronic $n=0$): the exact BO energy
error is resolved through sixth order into the positive diagonal BH term and the negative leading
omitted-state term in Eq.~\eqref{eq:geometricbudget}.  The arrow marks $v=2$, where the
exact lower inequality has already reversed.  The exact-error curve also contains
$O(\eta^8)$ and higher terms and therefore is not the algebraic sum of the two displayed
sixth-order contributions.}
\label{fig:budget}
\end{figure}
\FloatBarrier

The two models therefore answer the question posed at the start of the paper.  A large nuclear
mass helps because the derivative terms carry inverse powers of the nuclear mass.  But mass
alone does not determine the accuracy of a one-surface description.  The error also depends
on how rapidly the retained electronic state changes with nuclear position, the energy
separations from the electronic states that have been omitted, and the particular nuclear
state, which determines where in nuclear configuration space those local differences are
sampled.

\section{Conclusion}
\label{sec:conclusion}

The Born--Oppenheimer, Born--Huang, and exact-factorization descriptions begin from the same
molecular Schr\"odinger equation but do different things with the electronic degrees of
freedom.  The full Born--Huang expansion is an exact expansion of the molecular wavefunction
in electronic eigenstates calculated at fixed nuclear geometry.  Exact factorization is
another exact representation of the same molecular state.  At each nuclear configuration,
the Born--Huang coefficients are the components of the electronic vector in a chosen basis,
whereas exact factorization writes that same vector as its length times a normalized
direction.  BO and single-surface BH arise only after the Born--Huang expansion is restricted
to one electronic eigenstate.  BO omits the terms produced when that eigenstate is
differentiated with respect to the nuclear coordinates; single-surface BH retains the
diagonal part of those terms while omitting the coupling to the other electronic states.

The Fern\'andez oscillator makes these distinctions completely explicit because the full
molecular problem, the BO approximation, and the single-surface BH approximation can all be
solved analytically.  Fern\'andez already reported the Brattsev ground-state bracket; the exact
difference identities derived here establish its strict form for every $M>0$ and
$0<|\beta|<1$:
\[
\Ecal_{00}<E_{00}<E^A_{00},
\]
so the BO energy is a strict lower bound and the single-surface BH energy is a strict upper
bound.  This ordering is not a general rule for excited states.  Exact excited-state examples
show both $E_{nv}<\Ecal_{nv}<E^A_{nv}$ and $\Ecal_{nv}<E^A_{nv}<E_{nv}$.  Thus retaining
the diagonal BH correction does not by itself guarantee a more accurate excited-state
energy.

The derivative responsible for the BH correction also has a simple geometric interpretation.
For a normalized electronic eigenstate, the part of its derivative perpendicular to the
state measures how much the electronic state itself changes when the nuclear coordinate
changes.  The squared norm of this perpendicular derivative is the quantum-metric coefficient
$g$, and in one nuclear coordinate the single-surface BH correction is
\[
W=\frac{g}{2M}.
\]
The Fern\'andez model has the special property that both its electronic energy separations and
its metric coefficient are independent of the nuclear coordinate.  It therefore cannot show
how a changing electronic energy separation modifies these derivative terms.

The two-state model was introduced for that reason.  Its two electronic
eigenvalues have a separation
\[
\Delta(R)=2\sqrt{\kappa^2R^2+\lambda^2},
\]
which is smallest at $R=0$ and increases with $|R|$.  Decreasing this minimum separation
causes the electronic eigenvector to change more rapidly over a smaller interval of nuclear
coordinate: the metric peak grows as $\lambda^{-2}$ while becoming correspondingly narrower.
The total accumulated Fubini--Study distance remains $\pi/2$.  Thus changing the minimum
energy separation changes where the electronic state changes most rapidly, rather than the
total change between the two limiting electronic directions.

The same model also permits BO, single-surface BH, and EF nuclear calculations to be compared
directly.  Unlike the Fern\'andez model, its BH correction depends on $R$, so BO and
single-surface BH no longer differ merely by a constant energy shift.  For the ground-state
calculation considered here, the exact EF potential differs visibly from both one-surface
potentials near the region where the two electronic energies are closest.  Nevertheless,
\[
\bk{\chi_{\rm EF}}{\chi_{\rm BO}}=0.999984005,
\qquad
\bk{\chi_{\rm EF}}{\chi_{\rm BH}}=0.999999973.
\]
A visible local difference between potential-energy functions therefore need not produce a
comparably large difference in a particular nuclear state.

The exact Fern\'andez expansion explains why these observations have a common origin.
Its first correction to the BO energy is the positive diagonal BH term.  The next correction
contains the same derivative couplings, an additional division by the electronic energy
separation, and a factor determined by the nuclear state.  In Fern\'andez the electronic gaps
are constant, so increasing nuclear excitation exposes the nuclear-state dependence.  In the
two-state model the gap varies with nuclear position, so the additional inverse-gap dependence
can be seen directly.

The two benchmarks therefore give a more precise answer than the statement that BO becomes
accurate because nuclei are heavy.  Large nuclear mass suppresses the derivative terms, but
mass alone does not determine the error of a one-surface description.  The error also depends
on how rapidly the retained electronic state changes with nuclear position, the energy
separations from the electronic states that have been omitted, and where the nuclear
wavefunction has appreciable amplitude.  These quantities are distinct, and the simple
models used here allow their roles to be separated explicitly.

\appendix

\section{Closed-form calculations for the Fern\'andez benchmark}
\label{app:fernandez}

This appendix collects the algebra used in Sec.~\ref{sec:fernandez} and the quantities used
in the error decomposition of Sec.~\ref{sec:UQ}.  The normal-coordinate solution of a
bilinearly coupled oscillator is standard \cite{ZunigaBastidaRequena2017}; it is reproduced
so that every benchmark formula can be checked directly.

\subsection{Normal modes and the exact spectrum}

In the mass-weighted coordinates $y_1=x$, $y_2=\sqrt M\,X$ the kinetic energy has unit
mass in both directions.  Write the quadratic potential as
$\tfrac12\mathbf y^{\!\top}\mathcal K\mathbf y$, thereby defining the stiffness matrix
$\mathcal K$:
\begin{equation}
\mathcal K=
\begin{pmatrix}
1&\beta/\sqrt M\\[2pt]
\beta/\sqrt M&1/M
\end{pmatrix}.
\end{equation}
The original potential matrix
$\bigl[\begin{smallmatrix}1&\beta\\ \beta&1\end{smallmatrix}\bigr]$ is positive definite
precisely when $1-\beta^2>0$.  An orthogonal rotation
$\mathbf z=O^{\!\top}\mathbf y$ diagonalizes $\mathcal K$ without changing the flat
Laplacian.  Its eigenvalues are
\begin{equation}
k_\pm=
\frac{(M+1)\pm\sqrt{(M-1)^2+4M\beta^2}}{2M},
\end{equation}
so
\begin{equation}
\hat H=\sum_{s=\pm}\left(-\frac12\partial_{z_s}^2+\frac12k_s z_s^2\right)
\end{equation}
and Eq.~\eqref{eq:exactspectrum} follows.  The identities
\begin{equation}
k_++k_-=1+\frac1M,
\qquad
k_+k_-=\frac{1-\beta^2}{M}=\omega^2
\label{eq:kidentities}
\end{equation}
are the trace and determinant relations used in Proposition~\ref{prop:brattsev}.
If the rotation is written so that
$y_2=\sin\gamma\,z_++\cos\gamma\,z_-$, then
\begin{equation}
\tan2\gamma=\frac{2\beta\sqrt M}{M-1},
\qquad
\gamma=\frac{\beta}{\sqrt M}+O(M^{-3/2}).
\end{equation}

\subsection{Large-mass expansion and convergence}

Let $\eta=M^{-1/4}$, so that $M^{-1}=\eta^4$.  Expanding the exact frequencies gives
\begin{align}
\sqrt{k_+}&=1+\frac{\beta^2}{2(M-1)}+O(M^{-2}),\\
\sqrt{k_-}&=\omega\left[1-\frac{\beta^2}{2(M-1)}+O(M^{-2})\right].
\label{eq:kexpansions}
\end{align}
The slow-to-fast frequency ratio is therefore $O(\eta^2)$ and relative corrections begin at
$O(\eta^4)$.

With $M=\eta^{-4}$ in Eq.~\eqref{eq:kpm}, the nonzero singularities nearest
$\eta=0$ come from the zeros of the discriminant,
\begin{equation}
M_\pm=(1-2\beta^2)\pm2i\beta\sqrt{1-\beta^2}.
\end{equation}
Since
\begin{equation}
|M_\pm|^2=(1-2\beta^2)^2+4\beta^2(1-\beta^2)=1,
\end{equation}
the corresponding singularities lie on $|\eta|=1$.  The expansion around $\eta=0$
therefore converges for every $M>1$, as noted by Fern\'andez \cite{Fernandez1994}.

\subsection{Clamped problem, diagonal correction, and derivative couplings}

With $\xi=x+\beta X$,
\begin{equation}
\Hel=-\frac12\partial_\xi^2+\frac12\xi^2
     +\frac12(1-\beta^2)X^2.
\end{equation}
Hence $\phi_n(x;X)=\varphi_n(\xi)$ and Eq.~\eqref{eq:Un} follows.  The derivative is
$\partial_X\phi_n=\beta\varphi_n'(\xi)$.  The oscillator virial theorem
\cite{Fock1930}, or equivalently the ladder-operator identity
\begin{equation}
\varphi_n'=\frac{\sqrt n\,\varphi_{n-1}-\sqrt{n+1}\,\varphi_{n+1}}{\sqrt2},
\end{equation}
gives
\begin{equation}
W_n=\frac{\beta^2}{2M}\left(n+\frac12\right)
\end{equation}
and the only nonzero derivative couplings,
\begin{equation}
d_{n-1,n}=\beta\sqrt{\frac n2},
\qquad
d_{n+1,n}=-\beta\sqrt{\frac{n+1}{2}}.
\label{eq:ferndcouplings}
\end{equation}
Because $U_{n\pm1}-U_n=\pm1$,
\begin{equation}
g_n=\sum_{j\ne n}|d_{jn}|^2
=\beta^2\left(n+\frac12\right).
\label{eq:fern_g_app}
\end{equation}
For the lowest electronic state only $j=1$ contributes, and the corresponding gap is one:
\begin{equation}
\sum_{j\ne0}\frac{|d_{j0}|^2}{U_j-U_0}
=\frac{\beta^2}{2}.
\label{eq:fern_gapweighted_app}
\end{equation}
For $n=0$ this happens to equal $g_0$, but the two expressions contain different
information: the second includes an additional division by the electronic energy separation.

\subsection{Ground-state exact factorization}

The normal-mode ground state is Gaussian.  Since
$y_2=\sqrt M X=\sin\gamma\,z_++\cos\gamma\,z_-$ and
$\langle z_\pm^2\rangle=(2\omega_\pm)^{-1}$,
\begin{equation}
\langle X^2\rangle
=\frac{1}{2M}\left(\frac{\sin^2\gamma}{\omega_+}
                   +\frac{\cos^2\gamma}{\omega_-}\right).
\end{equation}
The nuclear amplitude is $\chi_{00}(X)\propto\exp(-aX^2/2)$ with
$a=[2\langle X^2\rangle]^{-1}$.  Inserting it into the exact nuclear equation gives
\begin{equation}
\varepsilon_{00}(X)-E_{00}
=\frac{1}{2M}\frac{\chi_{00}''}{\chi_{00}}
=\frac{a^2X^2-a}{2M},
\end{equation}
which is Eq.~\eqref{eq:fernEF} with $\omega_{\rm EF}=a/M$.

For completeness, the closed-form conditional quantities quoted in
Eqs.~\eqref{eq:fernEFmetricclosed} and \eqref{eq:fernEFHelclosed} follow directly from the
other row of the same orthogonal rotation,
\[
x=\cos\gamma\,z_+-\sin\gamma\,z_-.
\]
Writing $c=\cos\gamma$ and $s=\sin\gamma$, inversion gives
$z_+=cx+s\sqrt{M}X$ and $z_-=-sx+c\sqrt{M}X$.  At fixed $X$ the conditional Gaussian
therefore has width $A_{xx}=\omega_+c^2+\omega_-s^2$ and center
\[
x_0(X)=-\frac{\sqrt{M}(\omega_+-\omega_-)sc}{A_{xx}}X.
\]
For a normalized Gaussian of width $A_{xx}$,
$\langle (x-x_0)^2\rangle=(2A_{xx})^{-1}$, so
\[
\langle\partial_X\Phi|\partial_X\Phi\rangle
=\frac{A_{xx}}{2}\left(\frac{dx_0}{dX}\right)^2
=\frac{M(\omega_+-\omega_-)^2s^2c^2}{2A_{xx}}.
\]
At $X=0$ the conditional center vanishes and
$\hat H_{\rm el}(0)=(-\partial_x^2+x^2)/2$, hence
\[
\langle\Phi_{00,0}|\hat H_{\rm el}(0)|\Phi_{00,0}\rangle
=\frac{A_{xx}}4+\frac{1}{4A_{xx}}.
\]
Substitution of this conditional Gaussian into the general one-dimensional EF potential
formula, Eq.~\eqref{eq:EFgeometry1d}, reproduces Eq.~\eqref{eq:fernEF} identically.

\section{Fern\'andez's sixth-order correction in geometric variables}
\label{app:budget}

Let $\eta=M^{-1/4}$, so $M^{-1}=\eta^4$, and write
$\omega=\sqrt{1-\beta^2}\,\eta^2$.  Expanding the exact normal-mode frequencies gives
\begin{align}
\sqrt{k_+}
&=1+\frac{\beta^2}{2}\eta^4+O(\eta^8),\\
\sqrt{k_-}
&=\sqrt{1-\beta^2}\,\eta^2
 \left[1-\frac{\beta^2}{2}\eta^4+O(\eta^8)\right].
\end{align}
For electronic $n=0$,
\begin{align}
E_{0v}
&=\frac12+\left(v+\frac12\right)\sqrt{1-\beta^2}\,\eta^2
  +\frac{\beta^2}{4}\eta^4 \\
&\quad
 -\frac{\beta^2}{2}\left(v+\frac12\right)
  \sqrt{1-\beta^2}\,\eta^6+O(\eta^8).
\label{eq:exactexpansionapp}
\end{align}
Subtracting
\begin{equation}
\Ecal_{0v}=\frac12+\left(v+\frac12\right)
\sqrt{1-\beta^2}\,\eta^2
\end{equation}
gives Eq.~\eqref{eq:fernbudget}.

Equation~\eqref{eq:fern_g_app} gives $g_0=\beta^2/2$, so
\begin{equation}
\frac{g_0}{2M}=\frac{\beta^2}{4}\eta^4.
\end{equation}
The BO nuclear oscillator obeys
\begin{equation}
\bk{\chi_v^{(\omega)}}{\Tn\chi_v^{(\omega)}}
=\frac12\left(v+\frac12\right)\omega.
\end{equation}
Using Eq.~\eqref{eq:fern_gapweighted_app},
\begin{equation}
-\frac{2}{M}
\bk{\chi_v^{(\omega)}}{\Tn\chi_v^{(\omega)}}
\sum_{j\ne0}\frac{|d_{j0}|^2}{U_j-U_0}
=-\frac{\beta^2}{2}\left(v+\frac12\right)
\sqrt{1-\beta^2}\,\eta^6.
\end{equation}
Thus the two coefficients in the exact sixth-order expansion are the positive diagonal
metric term and the leading omitted term built from the same derivative couplings with one
additional inverse-gap weighting.  The $O(\eta^8)$ remainder follows directly from the
exact normal-mode expansion.

\end{document}